\documentclass[journal,comsoc]{IEEEtran}
\usepackage[T1]{fontenc}

\usepackage{cite}
\usepackage{afterpage}

\ifCLASSINFOpdf
  \usepackage[pdftex]{graphicx}
  \usepackage{epstopdf}
\else
  \usepackage[dvips]{graphicx}
\fi
\usepackage{amsmath}
\allowdisplaybreaks
\usepackage{amsthm}

\usepackage[cmintegrals]{newtxmath}

\usepackage{algorithmic}
\usepackage[algoruled, linesnumbered]{algorithm2e} 

\usepackage{array}

\ifCLASSOPTIONcompsoc
  \usepackage[caption=false,font=normalsize,labelfont=sf,textfont=sf]{subfig}
\else
\usepackage[caption=false,font=footnotesize]{subfig}
\fi
\usepackage{fixltx2e}

\usepackage{stfloats}
\fnbelowfloat
\usepackage{url}

\usepackage{graphicx}
\usepackage{setspace}
\usepackage[normalem]{ulem}
\usepackage{mathrsfs}
\usepackage{color}
\newtheorem{theorem}{Theorem}[section]

\newtheorem{conjec}{Conjecture}[section]

\newtheorem{lemma}{Lemma}[section]

\allowdisplaybreaks

\begin{document}
%
\title{Hybrid Beamforming with Selection for Multi-user Massive MIMO Systems \thanks{Part of this work was supported by the National Science Foundation.}}

\author{Vishnu~V.~Ratnam,~\IEEEmembership{Student~Member,~IEEE,}
\and Andreas~F.~Molisch,~\IEEEmembership{Fellow,~IEEE,} 
\and Ozgun~Y.~Bursalioglu,~\IEEEmembership{Member,~IEEE,}
and Haralabos~C.~Papadopoulos,~\IEEEmembership{Member,~IEEE} 
\thanks{V. V. Ratnam and A. F. Molisch are with the Ming Hsieh Department of Electrical Engineering, University of Southern California, Los Angeles, CA, 90089 USA (e-mail: \{ratnam, molisch\}@usc.edu)}
\thanks{O. Y. Bursalioglu and H. C. Papadopoulos are with Docomo Innovations, Palo Alto, CA, 94304 USA (e-mail: \{obursalioglu, hpapadopoulos\}@docomoinnovations.com)}
}

\IEEEtitleabstractindextext{%
\begin{abstract}
This work studies a variant of hybrid beamforming, namely, hybrid beamforming with selection (HBwS), as an attractive solution to reduce the hardware cost of multi-user Massive Multiple-Input-Multiple-Output systems, while retaining good performance. Unlike conventional hybrid beamforming, in a transceiver with HBwS, the antenna array is fed by an analog beamforming matrix with $\bar{L}$ input ports, where $\bar{L}$ is larger than the number of up/down-conversion chains $\bar{K}$. A bank of switches connects the instantaneously best $\bar{K}$ out of the $\bar{L}$ input ports to the up/down-conversion chains. The analog beamformer is designed based on average channel statistics and therefore needs to be updated only infrequently, while the switches operate based on instantaneous channel knowledge. HBwS allows use of simpler hardware in the beamformer that only need to adjust to the statistics, while also enabling the effective analog beams to adapt to the instantaneous channel variations via switching. This provides better user separability, beamforming gain, and/or simpler hardware than some conventional hybrid schemes. In this work, a novel design for the analog beamformer is derived and approaches to reduce the hardware and computational cost of a multi-user HBwS system are explored. 
In addition, we study how $\bar{L}$, the switch bank architecture, the number of users and the channel estimation overhead impact system performance.
\end{abstract}

}

\maketitle

\IEEEdisplaynontitleabstractindextext

%
\IEEEpeerreviewmaketitle

\section{Introduction} \label{introduction}
Massive Multiple-input-multiple-output (MIMO) systems, enabled by using antenna arrays with many elements at the transmitter and/or receiver, are viewed as a key enabler towards meeting the rising throughput demands in cellular systems \cite{Marzetta, Keytech_5G}. Such massive MIMO systems, while beneficial at micro-wave frequencies, are essential at millimeter (mm) wave frequency bands ($>30$ GHz) to compensate for the channel attenuation. 
It is predicted that future cellular systems will be equipped with antenna arrays having $100-1000$ antenna elements, at least at the base-station (BS) end. Although producing affordable large antenna arrays on a small footprint is already viable, the corresponding up/down-conversion chains, which include Analog-to-Digital Converters/ Digital-to-Analog Converters, filters, and mixers, are expensive and power hungry. 
This has motivated research on hybrid beamforming, which takes advantage of the directional nature of wireless channels \cite{Pedersen2000, molisch2014propagation, Akdeniz, Haneda2016}, to feed a large antenna array to fewer up/down-conversion chains. In this work, we focus primarily on frequency flat fading channels with such hybrid architectures at the BS. 
\subsection{Hybrid Beamforming} \label{subsec_hybrid_preprocessing}
In a BS with hybrid beamforming (also known as hybrid precoding/ preprocessing), an analog RF beamforming matrix, built from analog hardware like phase-shifters, is used to connect $N$ antenna elements to $\bar{K}$ up/down-conversion chains, where $\bar{K} < N$. This beamforming matrix exploits channel state information to form beams into the \emph{dominant} angular directions of each user's channel, thereby, utilizing the transmit power more effectively and providing some multi-user separation with fewer up/down-conversion chains. Since the analog hardware components are relatively cheap and consume less power than the up/down-converters and mixed signal components, this design leads to significant savings as compared to a full complexity transceiver, i.e., with $N$ up/down-conversion chains. 
The idea was first proposed in \cite{Molisch_VarPhaseShift, Sudarshan} and was further investigated for centimeter (cm) waves in \cite{Karamalis2006, venkateswaran2010analog} and for the mm-wave band in \cite{Ayach2012, Alkhateeb2013, Ayach_iCSI, Alkhateeb2014, Adhikary_JSDM, Alkhateeb2015}. Since then, numerous publications have followed suit with different architectures for the analog beamforming matrix under varying system models, and 3GPP is working on including them in the upcoming 5G cellular standard \cite{Docomo_3gpp_meeting}. An overview of the recent results is available in \cite{Alkhateeb2014_mag, Heath2016, Molisch_HP_mag}. 

The hybrid beamforming schemes can be broadly classified into two architectures based on the channel state information utilized to design the analog beamformer \cite{Molisch_HP_mag}. In one architecture called hybrid beamforming based on instantaneous channel state information (HBiCSI) \cite{Molisch_VarPhaseShift, Karamalis2006, Ayach_iCSI, Xu_iCSI, Alkhateeb2015, Yu2016, Sohrabi2016}, the beamforming matrix, and therefore the analog precoding beams, adapt to the instantaneous channel state information (iCSI), as illustrated in Fig.~\ref{Fig_illustrate_HPiCSI} for a single user case. Though this solution promises good performance\footnote{By system performance we refer to the capacity excluding the channel estimation overhead. \label{note1}}, iCSI across all the $N$ transmit antennas may be required leading to a large channel estimation overhead. While several approaches to reduce the estimation overhead have been proposed \cite{Adhikary_JSDM, Alkhateeb2014, Rial2016, Ratnam_ICC2018}, HBiCSI additionally imposes strict performance specification requirements on the analog hardware. This is because their parameters may have to be updated several times within each coherence time interval, which can be very short especially for mm-wave channels. 
In the other architecture called hybrid beamforming based on average channel state information (HBaCSI) \cite{Sudarshan, venkateswaran2010analog, Alkhateeb2013, Adhikary_JSDM, Liu2014, Park2017, Li2017}, the beamforming matrix adapts to the average channel state information (aCSI) i.e., the transmit/receive spatial correlation matrices, as illustrated in Fig.~\ref{Fig_illustrate_HPaCSI}. Since aCSI changes slowly, it can be acquired with a low overhead and the analog hardware parameters need to be updated infrequently. Additionally, iCSI is only needed in the channel sub-space spanned by the analog precoding beams, leading to a significant reduction in the channel estimation overhead. 
Despite these benefits, the $\text{performance}^{1}$ may be worse than HBiCSI since the analog beams can only span a fixed channel subspace of dimension $\bar{K}$, which does not adapt to iCSI \cite{Adhikary_JSDM}. The performance gap may be especially large if the dimension of the dominant channel subspace\footnote{It represents the channel subspace at the transmitter along which a significant portion of the channel power is concentrated. Such a notion is quite common for massive MIMO and is used in several proposed schemes like Joint Space Division Multiplexing \cite{Adhikary_JSDM}. \label{note_dominant}} is much larger than $\bar{K}$, which is possible both at microwave \cite{TR36_873} and mm-wave \cite{Akdeniz} frequencies. 
While there has been a significant amount of work on HBiCSI and HBaCSI under different system models and constraints, there is limited work on bridging the performance gap between the two. 
\begin{figure}[!h]
\centering
\subfloat[HBiCSI]{\label{Fig_illustrate_HPiCSI}\includegraphics[width= 0.28\textwidth]{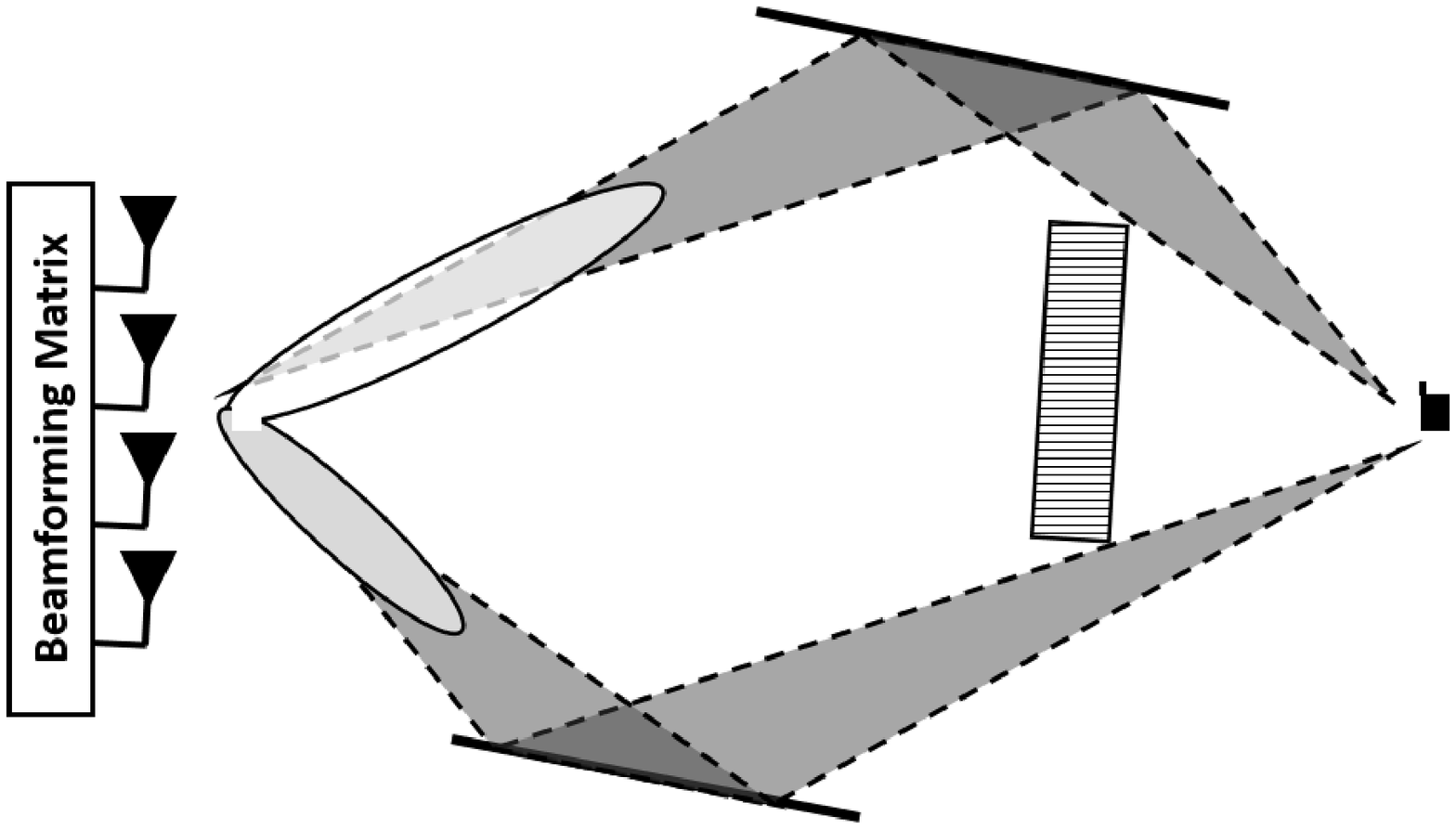}} \hspace{5mm}
\subfloat[HBaCSI]{\label{Fig_illustrate_HPaCSI}\includegraphics[width= 0.28\textwidth]{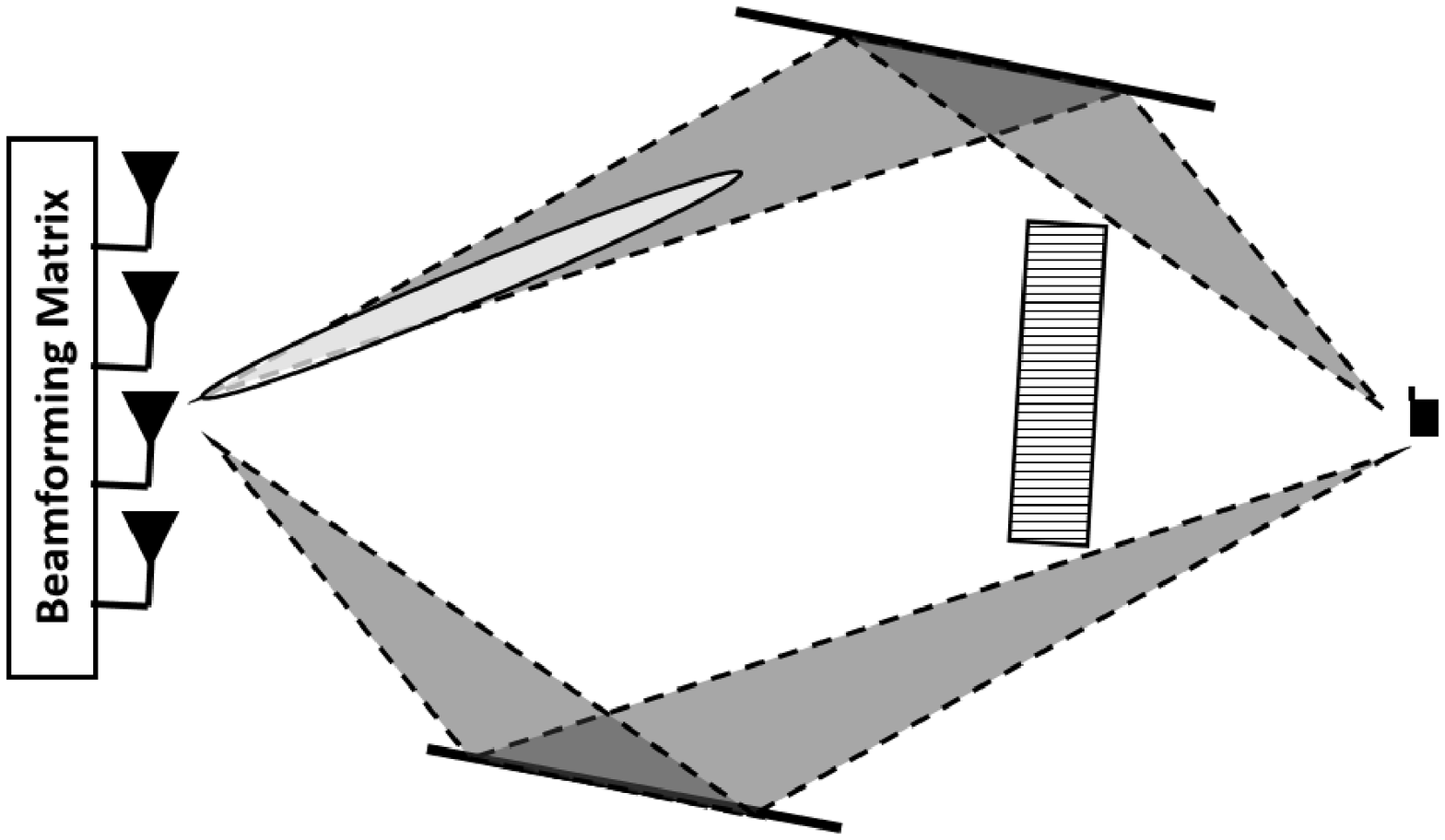}} \hspace{5mm}
\subfloat[HBwS]{\label{Fig_illustrate_HPwS}\includegraphics[width= 0.33\textwidth]{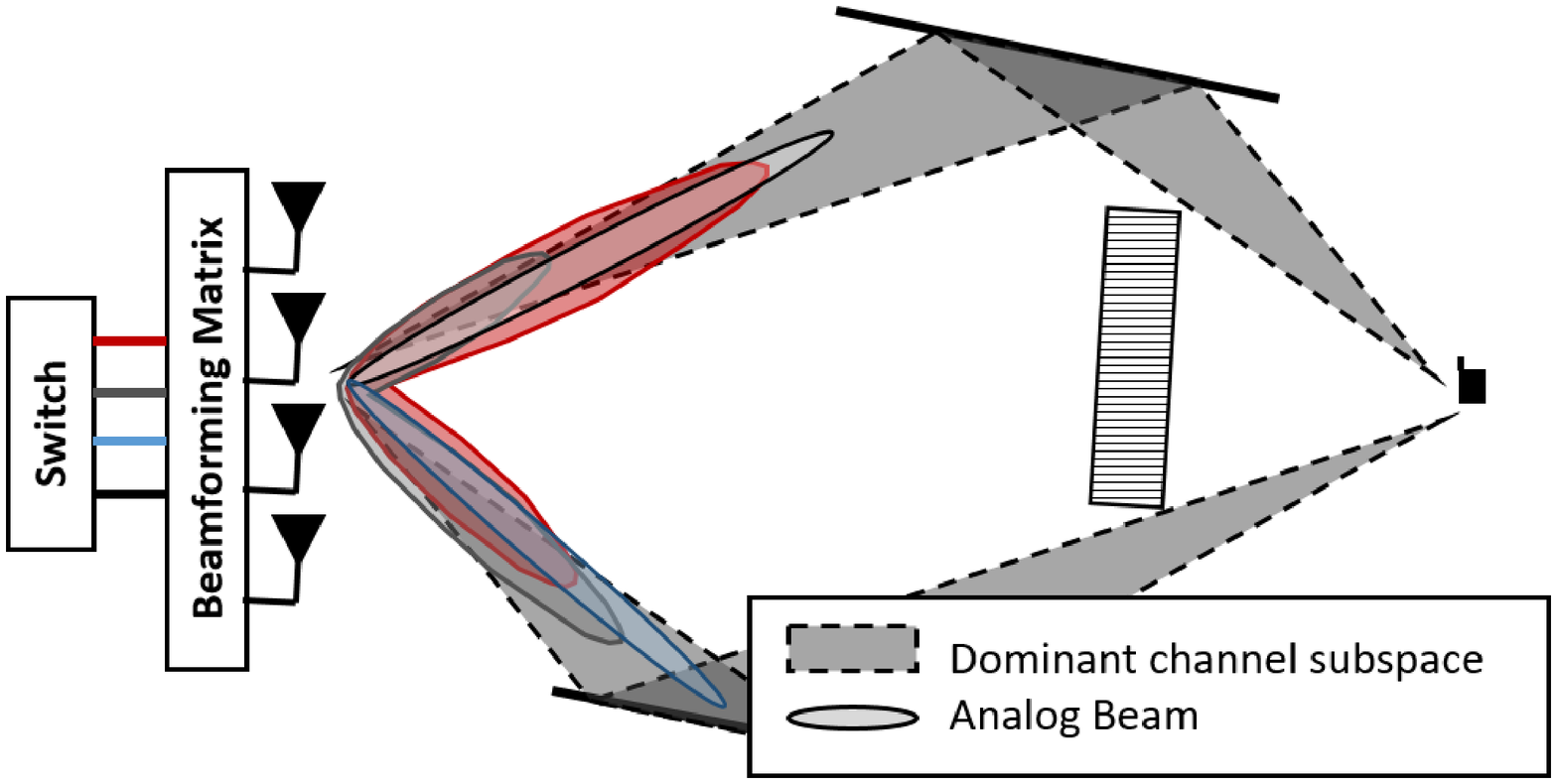}}
\caption{An illustration of the different hybrid beamforming schemes at a transmitter with one up-conversion chain and a single antenna receiver.}
\label{fig_illustrate_HPschemes}
\end{figure}
\subsection{Hybrid Beamforming with Selection} \label{subsec_HBwS_intro}
Even with a small $\bar{K}$ and an aCSI based analog beamformer, it is possible to adapt the transmit analog precoding beams to iCSI via the use of selection techniques \cite{Sudarshan}. 
By using additional analog hardware, several possible options for the analog precoding beams can be provided to span the dominant channel subspace, as illustrated in Fig.~\ref{Fig_illustrate_HPwS}. By dynamically switching to the ``best" beams for each channel realization, we may obtain $\text{performance}^{1}$ comparable to HBiCSI. 

In this paper we study a generalization of this approach, namely, hybrid beamforming with selection (HBwS), as a solution to achieve $\text{performance}^{1}$ comparable to HBiCSI, while still retaining some benefits of HBaCSI i.e., infrequent update of analog hardware parameters and low channel estimation overhead. 
The block diagram of a transmitter (TX) with HBwS is given in Fig.~\ref{Fig_blockdiag_sel}. Here, we again have an analog beamforming matrix that is connected to the antennas. However, unlike conventional hybrid beamforming, the number of input ports for the beamforming matrix ($\bar{L}$) is larger than the number of available up-conversion chains ($\bar{K}$). 
\begin{figure}[!h]
\centering
\subfloat[Block Diagram]{\label{Fig_blockdiag_sel}\includegraphics[width= 0.45\textwidth]{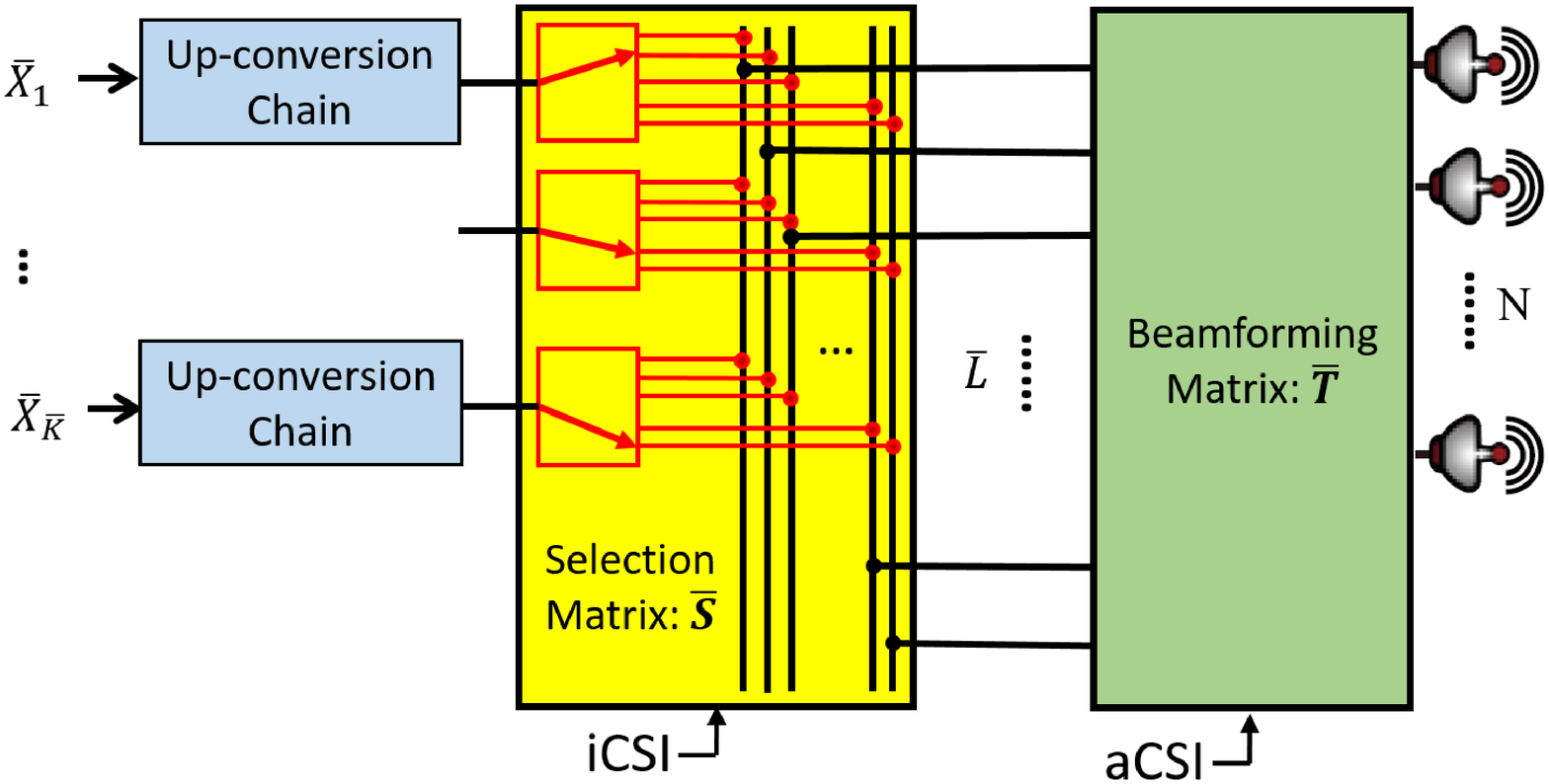}} \hspace{5mm} 
\subfloat[User layout]{\label{Fig_scenario_illustrate}\includegraphics[width= 0.3\textwidth]{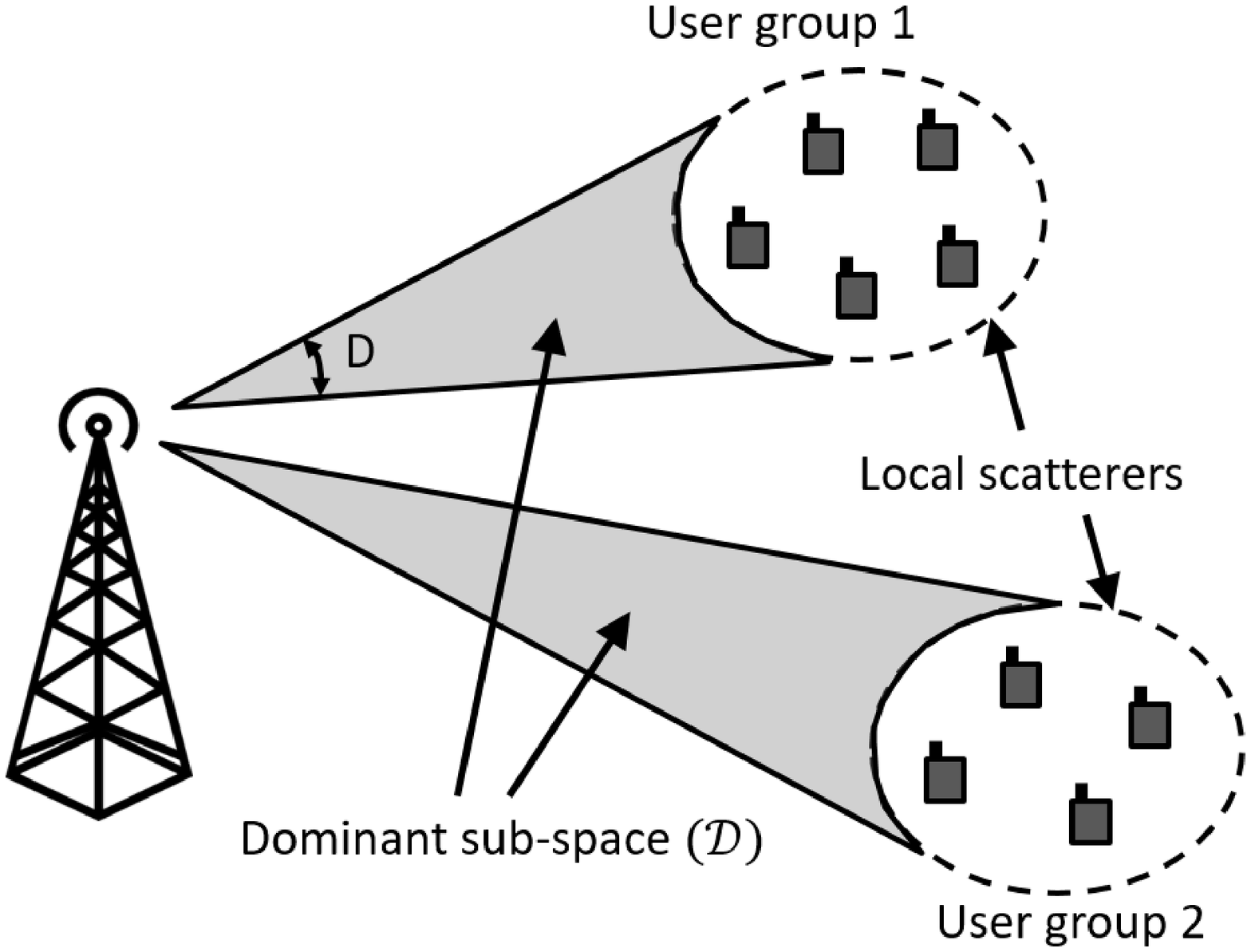}} 
\caption{Illustration of: (a) a Block diagram of Hybrid Beamforming with Selection at the TX (b) a sample user layout}
\label{Fig_Illustrate_HPwS}
\end{figure}
This matrix is preceded by a bank of $\bar{K}$ one-to-many RF switches, each of which connects one up-conversion chain to one of several input ports. Note that each connection of the $\bar{K}$ up-conversion chains to $\bar{K}$ out-of-the $\bar{L}$ input ports corresponds to a distinct analog precoding beam in Fig.~\ref{Fig_illustrate_HPwS}. 
While the beamforming matrix is designed based on aCSI, the switches exploit iCSI to optimize this $\bar{K}$ out-of-the $\bar{L}$ input port selection. 
Since the beamformer uses only aCSI, HBwS is more resilient to the transient response of phase-shifters \cite{Romanofsky2004} than HBiCSI, thereby easing performance specification requirements on them.\footnote{Adapting to iCSI involves changing the precoding beam multiple times within a coherence interval (see Section \ref{sec_chan_est_overhead}).} 
The premise for this design is that unlike phase shifters, RF switches are cheap, have low insertion loss and can be easily designed to switch quickly based on iCSI \cite{Rial2016, Sudarshan, Romanofsky2004, Yeh2017, Firouzjael2010, Schmid2014}.
Since $\bar{L} > \bar{K}$ and switches adapt the \emph{effective} precoding beams to iCSI, the analog precoding can span a larger channel subspace and a larger beamforming gain can be achieved in comparison to HBaCSI. 
Furthermore, due to its superior beam-shaping capabilities, HBwS provides better user separation than HBaCSI in a multi-user system. 
On the downside, since the analog beams span a larger channel subspace, the channel estimation overhead for HBwS may be larger than for HBaCSI. 
But by designing the beamformer carefully, the overhead can be made significantly lower than for full channel estimation. Additionally, since the beamformer has a size of $N \times \bar{L}$, as opposed to $N \times \bar{K}$ for HBaCSI, a larger number of analog components are required for HBwS. The estimation overhead and techniques to reduce the number of analog components are discussed in more detail in Sections \ref{sec_chan_est_overhead} and \ref{sec_hardware_cost}. 

The simplest type of HBwS is antenna selection \cite{Paulraj, Molisch_mag, Magazine_nosratinia}, where the analog beamforming matrix is omitted. Introducing a beamforming stage provides additional beamforming gains in correlated channels, and therefore, several designs for the beamforming matrix using phase-shifters \cite{Molisch_FFTshift, Sudarshan} or lens antennas \cite{Zeng2014, Gao2017} have been proposed. More recently, antenna selection has also been explored with regard to cost, power consumption and channel estimation overhead \cite{Rial2016}. However in most of the prior works, the beamforming matrix offers orthogonal beam choices and the number of input ports ($\bar{L}$) equals the number of transmit antennas ($N$), i.e., the beams span the whole channel dimension. Though some of these designs \cite{Molisch_VarPhaseShift, Sudarshan}, can be extended to the case of $\bar{L} < N$, these designs are inferior, especially in spatially sparse channels \cite{Vishnu_ICC2017}. 
A more generic design for the beamforming matrix, with $\bar{L} \neq N$ and possibly non-orthogonal columns (i.e., offering non-orthogonal beam choices), was proposed by us in \cite{Vishnu_ICC2017} for a single user multiple-input-single-output scenario and shown to provide improved performance. This work extends this design to a multi-user MIMO scenario, while also accounting for the impact of the switch bank architecture. Furthermore, we investigate the hardware implementation cost of HBwS. 
The contributions of this paper are as follows:
\begin{enumerate}
\item We propose a generic architecture of HBwS for low complexity multi-user MIMO transceivers, wherein the beamforming matrix may be a rectangular matrix i.e., $\bar{L} \neq N$, with non-orthogonal columns. 
\item For a channel with isotropic scattering within the subspace spanned by the beamformer,
we show that a beamformer that maximizes a lower bound to the system sum capacity with Dirty-paper coding can be obtained as a solution to a coupled Grassmannian subspace packing problem. 
\item We find a good sub-optimal solution to this packing problem and propose algorithms to improve it. 
\item We propose a two-stage architecture for the beamformer and find a family of ``good" switch positions, which help reduce the computational and hardware cost of HBwS while retaining good performance. 
\item An extension of the beamformer design to channels with anisotropic scattering is also explored.
\end{enumerate}
Some other works have explored a different hybrid architecture, involving switches between the antennas and the analog beamforming matrix \cite{Alkhateeb2016}. Unlike here, the purpose there is to reduce the number of analog components in the beamforming matrix. 

The organization of this paper is as follows: the general assumptions and the channel model are discussed in Section \ref{sec_chan_model}; the  sum capacity maximizing beamforming matrix design problem is formulated in Section \ref{sec_prbm_formulate}; the search-space for the optimal beamformer is characterized in Section \ref{sec_restrict_search_space} and a closed-form lower bound to the sum capacity is explored in Section \ref{sec_approx_obj_fn}; a good beamformer design and algorithms for further improving upon it are proposed in Section \ref{sec_pp_design}; strategies to reduce the hardware implementation cost for HBwS are discussed in Section \ref{sec_hardware_cost}; the channel estimation overhead for HBwS is discussed in Section \ref{sec_chan_est_overhead}; the simulation results are presented in Section \ref{sec_sim_results}; the extension to anisotropic channels is considered in Section \ref{sec_anisotropy}; and finally, the conclusions are summarized in Section \ref{sec_conclusions}.

Notation used in this work is as follows: scalars are represented by light-case letters; vectors by bold-case letters; matrices are represented by capitalized bold-case letters and sets and subspaces are represented by calligraphic letters. Additionally, $\mathbf{a}_i$ represents the $i$-th element of a vector $\mathbf{a}$, $|\mathbf{a}|$ represents the $L_2$ norm of a vector $\mathbf{a}$, $\mathbf{A}_{i,j}$ represents the $(i,j)$-th element of a matrix $\mathbf{A}$, ${[\mathbf{A}]}_{{\rm c}\{i\}}$ and ${[\mathbf{A}]}_{{\rm r}\{i\}}$ represent the i-th column and row vectors of matrix $\mathbf{A}$ respectively, ${\|\mathbf{A}\|}_F$ represents the Frobenius norm of a matrix $\mathbf{A}$, $\mathbf{A}^{\dag}$ is the conjugate transpose of a matrix $\mathbf{A}$, $\mathcal{P}\{\mathbf{A}\}$ represents the subspace spanned by the columns of a matrix $\mathbf{A}$, $|\mathbf{A}|$ represents the determinant of a matrix $\mathbf{A}$ and $|\mathcal{A}|$ represents the cardinality of a set $\mathcal{A}$ or dimension of a space $\mathcal{A}$. Also, ${}^N {\rm C}_k = \frac{N!}{(N-k)! k!}$ where $N!$ is the factorial of $N$, $\,{\buildrel d \over =}\,$ is equivalence in distribution, $\succeq 0$ implies a positive semi-definite constraint, $\mathbb{E}\{\}$ represents the expectation operator, $\mathbb{P}$ is the probability operator, $\mathbb{I}_i$ and $\mathbb{O}_{i,j}$ are the $i \times i$ and $i \times j$ identity and zero matrices respectively, and $\mathbb{R}$ and $\mathbb{C}$ represent the field of real and complex numbers.

\section{General Assumptions and Channel model} \label{sec_chan_model}
We consider the downlink of a single cell system, having one BS with $N$ antennas ($N \gg 1$) and implementing HBwS, and multiple users, modeled as a multi-user massive MIMO broadcast channel. The presented results can also be extended to the uplink multiple-access channel (MAC) with HBwS at the BS. Similar to the system model in \cite{Adhikary_JSDM}, we assume that the users can be spatially divided into user groups, with common intra-group spatial channel statistics and orthogonal channels across the groups, as illustrated in Fig.~\ref{Fig_scenario_illustrate}. If such a grouping for all users is not possible, a suitable user selection algorithm can be used \cite{Adhikary2014}. We further assume that the BS transmit resources, such as the RF chains, TX power, switches and analog hardware are split among the different user groups based on the average channel statistics, via an aCSI based resource sharing algorithm.\footnote{While iCSI based inter-group resource sharing may potentially improve performance, it may pose stringent requirements on the system hardware and increase system complexity.} While such algorithms already exist for HBaCSI \cite{Liu2014, Wang2015, Li2017}, extending them for HBwS is beyond the scope of this paper. Since the different user groups have orthogonal channels and the resources are split among the groups based on aCSI, the transmission to each user group can be treated independently. Therefore, without loss of generality, we restrict the analysis to a single representative user group with $M_1$ users. The BS allocates $K$ up-conversion chains to this user group ($K \leq N$). The portion of the TX RF analog beamforming matrix allocated to the user group $\mathbf{T}$, has a dimension of $N \times L$ and a corresponding sub-set of switches, denoted by a selection matrix $\mathbf{S}$, are used to connect $K$ out-of-the $L$ input ports of the beamforming matrix to the $K$ up-conversion chains. Note that $K \leq \bar{K}$, $L \leq \bar{L}$, $\mathbf{T}$ is a sub-matrix of $\bar{\mathbf{T}}$ and $\mathbf{S}$ is a sub-matrix of $\bar{\mathbf{S}}$, where the right hand sides refer to the total resources at the BS (see Section \ref{subsec_HBwS_intro} and Fig.~\ref{Fig_blockdiag_sel}). The $M_1$ receivers in the group have $M_2$ antennas and $M_2$ down-conversion chains each. We further define $M \triangleq M_1 M_2$ and assume that $K \geq M$. 
We consider a narrow-band system with a frequency flat and temporally block fading channel. Under these assumptions, the downlink baseband  received signal at user $m$, for a given selection matrix $\mathbf{S}$, can be expressed as: 
\begin{subequations}
\begin{eqnarray}
\mathbf{y}_m(\mathbf{S}) &=& \sqrt{\rho} \widetilde{\mathbf{H}}_m \mathbf{T} \mathbf{S} \mathbf{x} + \mathbf{n}_m \label{eqn_downlink_y1} \\
&=& \sqrt{\rho} \widetilde{\mathbf{H}}_m \mathbf{T} \mathbf{S} \mathbf{G u} + \mathbf{n}_m \label{eqn_downlink_y2}
\end{eqnarray}
\end{subequations}
where $\mathbf{y}_m(\mathbf{S})$ is the $M_2 \times 1$ received signal vector at user $m$, $\rho$ is the mean receive signal-to-noise ratio (SNR), $\widetilde{\mathbf{H}}_m$ is the $M_2 \times N$ downlink channel matrix for user $m$, $\mathbf{S}$ is a $L \times K$ sub-matrix of the identity matrix $\mathbb{I}_{L}$ - formed by picking $K$ out-of-the $L$ columns, $\mathbf{x}$ is the $K \times 1$ sub-matrix of the $\bar{K} \times 1$ digitally precoded transmit vector $\bar{\mathbf{x}}$ corresponding to the $K$ allocated up-conversion chains and $\mathbf{n}_m \sim \mathcal{CN}(\mathbb{O}_{M_2\times M_2},\mathbb{I}_{M_2})$ is the normalized additive white Gaussian noise observed at user $m$. Without loss of generality we define $\mathbf{u} \triangleq \mathbf{G}^{-1} \mathbf{x}$, where $\mathbf{G}$ is a $K\times K$ a full-rank matrix that ortho-normalizes the columns of $\mathbf{T}\mathbf{S}$ i.e., $\mathbf{G}^{\dag} \mathbf{S}^{\dag} \mathbf{T}^{\dag} \mathbf{T}\mathbf{S} \mathbf{G} = \mathbb{I}_K$. Here we implicitly assume that $\mathbf{T} \mathbf{S}$ has linearly independent columns for each $\mathbf{S}$. 
The transmit power constraint can then be expressed as:
\begin{subequations}
\begin{eqnarray} 
\text{tr}\{\mathbf{T S} \mathbb{E}_{\mathbf{x}}\{\mathbf{x} \mathbf{x}^{\dag}\} \mathbf{S}^{\dag} \mathbf{T}^{\dag} \} \leq 1 \label{eqn_TX_pow_constraint1} \\
\Rightarrow \mathbb{E}_{\mathbf{u}}\{\mathbf{u}^{\dag} \mathbf{u} \} \leq 1. \label{eqn_TX_pow_constraint}
\end{eqnarray}
\end{subequations}
Note that the ortho-normalization matrix $\mathbf{G}$ is defined only to simplify \eqref{eqn_TX_pow_constraint1} and does not constitute the entire digital base-band precoding, a part of which may still exist in $\mathbf{u}$. 

The channel is assumed to contain both a large scale fading as well as a small scale fading component. The small scale fading statistics are assumed to be Rayleigh in amplitude and doubly spatially correlated (both at transmitter and receiver end). As illustrated in Fig.~\ref{Fig_scenario_illustrate}, the users in a group are close enough to share the same set of local scatterers, but are sufficient wavelengths apart to undergo independent and identically distributed (i.i.d.) small scale fading. Therefore, we assume that the channels to the different users are independently distributed and follow the widely used Kronecker correlation model \cite{Kronecker_validity} with a common transmit spatial correlation matrix $\mathbf{R}_{{\rm tx}}$ but individual receive correlation matrices $\mathbf{R}_{{\rm rx},m}$, respectively. Let $\mathbf{R}_{\rm tx} = \mathbf{E}_{\rm tx} \mathbf{\Lambda}_{\rm tx} \mathbf{E}_{\rm tx}^{\dag}$ be the eigen-decomposition of the transmit correlation matrix such that the diagonal elements of $\mathbf{\Lambda}_{\rm tx}$ are arranged in descending order of magnitude. 
Under these conditions, the channel matrices can be expressed as:
\begin{eqnarray}
\widetilde{\mathbf{H}}_m = \mathbf{R}_{{\rm rx},m}^{1/2} \mathbf{H}_m {[\boldsymbol{\Lambda}_{\rm tx}]}^{1/2} {[\mathbf{E}_{\rm tx}]}^{\dag}, \label{eqn_channel_KL_exapnsion}
\end{eqnarray}
where $\mathbf{H}_m$ is an $M_2 \times N$ matrix with i.i.d. $\mathcal{CN}(0, 1)$ entries. Without loss of generality, the mean pathloss for the user group is included into $\rho$ and any user specific large scale fading components are included in $\mathbf{R}_{{\rm rx}, m}$. 

The BS is assumed to have perfect knowledge of the aCSI metrics $\{\mathbf{R}_{\rm tx}, \mathbf{R}_{\rm rx,1}, ..., \mathbf{R}_{\rm rx,M_1}\}$, which is used to update the analog beamforming matrix. The BS is also assumed to have perfect knowledge of the effective iCSI channels for the user group $\{\widetilde{\mathbf{H}}_{1} \mathbf{T}, ..., \widetilde{\mathbf{H}}_{M_1} \mathbf{T}\}$, which is used to update the selection matrix $\mathbf{S}$. Finally, each user $m$ is assumed to know its effective channel after picking $\mathbf{T}$ and $\mathbf{S}$, i.e., $\widetilde{\mathbf{H}}_{m} \mathbf{T S G}$. The corresponding channel estimation overhead is discussed later in Section \ref{sec_chan_est_overhead}. 
\section{Problem Formulation} \label{sec_prbm_formulate}
We consider a generic switching architecture, where $\mathcal{S} \triangleq \{\mathbf{S}_1,.., \mathbf{S}_{|\mathcal{S}|}\}$ denotes the set of all feasible selection matrices. Note that depending on the switch bank architecture, this set, referred to as the switch position set, may not involve all the ${L \choose K}$ choices. For each $\mathbf{S}_i \in \mathcal{S}$, let $\mathbf{G}_i$ be the corresponding orthogonalization matrix in \eqref{eqn_downlink_y2} i.e., $\mathbf{G}_i^{\dag} \mathbf{S}_i^{\dag} \mathbf{T}^{\dag} \mathbf{T}\mathbf{S}_i \mathbf{G}_i = \mathbb{I}_K$. Although $\mathbf{G}_i$ is also a function of $\mathbf{T}$, this dependence is not explicitly shown for ease of representation. For a given selection matrix $\mathbf{S}_i$, note that the downlink channel described in Section \ref{sec_chan_model} is a broadcast channel with effective channel matrices $\widetilde{\mathbf{H}}_m \mathbf{T} \mathbf{S}_i \mathbf{G}_i$ for each user $m$. 
Therefore using uplink-downlink duality \cite{Vishwanath_duality}, the ergodic sum capacity achievable using Dirty-Paper coding \cite{Costa_DPC, DPC_2006} can be expressed as:
\begin{flalign}
& C(\mathbf{T}) = \mathbb{E}_{\widetilde{\mathbf{H}}} \Bigg\{ \max_{1 \leq i \leq |\mathcal{S}|, \{\mathbf{P}_1,..,\mathbf{P}_{M_1}\}} \Bigg( & \nonumber \\
& \qquad \quad \ \ \log \bigg| \mathbb{I}_{K} \!+\!\! \sum_{m=1}^{M_1}\! \rho \mathbf{G}^{\dag}_i \mathbf{S}_i^{\dag} \mathbf{T}^{\dag} \widetilde{\mathbf{H}}^{\dag}_m \mathbf{P}_m  \widetilde{\mathbf{H}}_m \mathbf{T} \mathbf{S}_i \mathbf{G}_i \bigg| \Bigg) \Bigg\} \!\!\!\!\!\!\! & \label{eqn_cap_MU} \\
& \quad \qquad \text{subject to:} \ \mathbf{P}_m \succeq 0, \sum_{m=1}^{M_1} \text{Tr}\{\mathbf{P}_m\} \leq 1, & \nonumber
\end{flalign}
where $\mathbf{P}_m$ represents the $M_2 \times M_2$ dual-uplink MAC transmit covariance matrix at user $m$ and we define $\widetilde{\mathbf{H}} = {\big[ \begin{array}{cccc} \widetilde{\mathbf{H}}^{\dag}_1 & \widetilde{\mathbf{H}}^{\dag}_2 & \hdots & \widetilde{\mathbf{H}}^{\dag}_{M_1} \end{array} \big]}^{\dag}$. 
Although \eqref{eqn_cap_MU} is convex in $\{\mathbf{P}_m\}$ for each $i$, the optimal covariance matrices $\{\mathbf{P}_m\}$ are not known in closed form. Therefore, we rely on a sub-optimal solution: $\mathbf{P}_m = \frac{1}{M}\mathbb{I}_{M_2}$, which is optimal under a large SNR if $\widetilde{\mathbf{H}}$ has a full row-rank and $K \geq M$ \cite{Caire_DPC, Lee2007}.\footnote{The SNR here is including the analog beamforming gain, and additionally the users in a group have a comparable signal strength. Therefore these conditions may be usually satisfied if $\mathbf{R}_{{\rm rx}, m}$ has a full rank $\forall m$.} 
Using this solution in \eqref{eqn_cap_MU} and applying the Sylvester's determinant identity \cite{Sylvesters_identity}, we obtain a tractable sum capacity lower bound:
\begin{flalign} \label{eqn_cap_MU_nearopt}
&\underbar{C}(\mathbf{T}) \triangleq \mathbb{E}_{\widetilde{\mathbf{H}}} \left\{ \max_{1 \leq i \leq |\mathcal{S}|} \log \left| \mathbb{I}_{M} + \frac{\rho}{M} \widetilde{\mathbf{H}} \mathbf{T} \mathbf{S}_i \mathbf{G}_i \mathbf{G}^{\dag}_i \mathbf{S}_i^{\dag} \mathbf{T}^{\dag} \widetilde{\mathbf{H}}^{\dag} \right| \right\}. &
\end{flalign}
Since we assume the use of the capacity optimal dirty-paper coding instead of linear pre-coding, a base-band digital beamforming matrix does not show up in \eqref{eqn_cap_MU_nearopt}. We shall henceforth refer to $\underbar{C}(\mathbf{T})$ as the hSNR sum capacity, recalling that $\underbar{C}(\mathbf{T}) \leq C(\mathbf{T})$ in general, with equality in the high SNR regime. 
From \eqref{eqn_channel_KL_exapnsion} and the fact that $\widetilde{\mathbf{H}}_m$ are independent for all $m$, we can express:
\begin{eqnarray} \label{eqn_channel_KL_expansion_full}
\widetilde{\mathbf{H}} = \mathbf{R}_{{\rm rx}}^{1/2} \mathbf{H} {[\boldsymbol{\Lambda}_{\rm tx}]}^{1/2} {[\mathbf{E}_{\rm tx}]}^{\dag}, 
\end{eqnarray}
where $\mathbf{R}_{{\rm rx}}$ is a block-diagonal matrix with the $m$-th diagonal block being $\mathbf{R}_{{\rm rx},m}$ and $\mathbf{H} = {\big[ \begin{array}{cccc} \mathbf{H}^{\dag}_1 & \mathbf{H}^{\dag}_2 & \hdots & \mathbf{H}^{\dag}_{M_1} \end{array} \big]}^{\dag}$ is a $M \times N$ matrix with i.i.d. $\mathcal{CN}(0, 1)$ entries. The primary goal of this work is to find the analog beamformer $\mathbf{T}$ that maximizes the lower bound in \eqref{eqn_cap_MU_nearopt}, i.e.,:
\begin{eqnarray}
\mathbf{T}_{\rm opt} = {\rm argmax}_{\mathbf{T} \in \mathbb{C}^{N\times L}} \{ \underbar{C}(\mathbf{T}) \} \label{eqn_prob_formulate} \\
\text{subject to: } |\mathcal{P}\{\mathbf{T}\}| \leq D, \qquad \nonumber
\end{eqnarray}
where $\mathcal{P}\{\mathbf{T}\}$ represents the sub-space spanned by the columns of $\mathbf{T}$, $D$ is a bound on the dimension of this subspace, $\mathbf{T}_{\rm opt}$ is designed based on the knowledge of the aCSI statistics: $\mathbf{R}_{\rm tx}$ and $\mathbf{R}_{\rm rx}$ and, with slight abuse of notation, ${\rm argmax}\{\}$ refers to any one of the (possibly many) maximizing arguments. Such a bound on $\mathcal{P}\{\mathbf{T}\}$ is required to limit the channel estimation overhead, as shall be shown in Section \ref{sec_chan_est_overhead}. The optimal beamformer that maximizes an objective of the form $f(D)\underbar{C}(\mathbf{T})$, where $f(\cdot)$ is any non-increasing function, can then be found by simply performing a line search over $D \in \{1,...,N\}$. 
Note that \eqref{eqn_prob_formulate} does not involve any magnitude or phase restrictions on the elements of the analog beamforming matrix $\mathbf{T}$. Such an unrestricted beamformer, while serving as a good reference for comparison, may also provide intuition for designing beamformers with unit magnitude and discrete phase constraints. 
As shall be seen later, an exact solution to \eqref{eqn_prob_formulate} is intractable and we shall therefore restrict ourselves to a good sub-optimal solution, that only requires knowledge of $\mathbf{R}_{\rm tx}$.

\subsection{Connections to limited-feedback precoding} \label{sec_connections_LFP}
Note that HBwS is an example of a restricted precoded system \cite{Vishnu_jrnl1}. In fact, by considering the precoding matrices for the different switch positions $\{\mathbf{T} \mathbf{S}_i \mathbf{G}_i| \mathbf{S}_i \in \mathcal{S}\}$ as entries of a codebook, the single user 
case ($M_1=1$) can be interpreted as a type of limited-feedback unitary precoding \cite{Love_SMux, Review_love}. 
However, in contrast to conventional limited-feedback precoding, the HBwS codebook entries $\{\mathbf{T} \mathbf{S}_i \mathbf{G}_i| \mathbf{S}_i \in \mathcal{S}\}$ are coupled, as they are generated from the columns of the same beamforming matrix $\mathbf{T}$. As a result, good codebook designs for limited-feedback unitary precoding \cite{Love_SMux, Love_beamforming_correlated, Raghavan_rot_scale, Raghavan_RVQ} cannot be directly extended to find good designs for $\mathbf{T}$.

\section{Transforming the search space} \label{sec_restrict_search_space}
Notice that in \eqref{eqn_prob_formulate}, search for $\mathbf{T}_{\rm opt}$ is over all possible $N \times L$ complex matrices with $|\mathcal{P}\{\mathbf{T}\}| \leq D$. Without loss of generality, such a beamformer can be expressed as $\mathbf{T} = \mathbf{E} \hat{\mathbf{T}}$ where $\mathbf{E} \in \mathbb{C}^{N\times D}$ and $\hat{\mathbf{T}} \in \mathbb{C}^{D \times L}$. In this section, we reduce this search space by getting rid of some sub-optimal and redundant solutions. We first state the following theorem:
\begin{theorem} \label{Th_restrict_dom_eigspace}
If ${\rm rank}\{\mathbf{R}_{\rm tx}\} \leq D$, there exists an optimal solution $\mathbf{T}_{\rm opt}$ to \eqref{eqn_prob_formulate} such that, $\mathbf{T}_{\rm opt} = \mathbf{E}^D_{\rm tx} \hat{\mathbf{T}}$, where $\hat{\mathbf{T}} \in \mathbb{C}^{D \times L}$ and $\mathbf{E}^D_{\rm tx} = {[\mathbf{E}_{\rm tx}]}_{{\rm c}\{1:D\}}$ is the $N \times D$ principal sub-matrix of $\mathbf{E}_{\rm tx}$ corresponding to the $D$ largest eigenvalues.
\end{theorem}
\begin{proof}
See Appendix \ref{appdix1}.
\end{proof}
Note that such a low rank $\mathbf{R}_{\rm tx}$ may be often experienced in massive MIMO systems both at cm and mm wave frequencies. We also conjecture that:
\begin{conjec} \label{conj_dom_eigspace}
Theorem \ref{Th_restrict_dom_eigspace} holds for any ${\rm rank}\{\mathbf{R}_{\rm tx}\}$.
\end{conjec}
Intuitively, this conjecture states that if $\mathcal{P}\{\mathbf{T}\}$, and therefore the analog precoding beams, should lie in a channel subspace of dimension $D$, then it should be the dominant $D$ dimensional channel sub-space $\mathcal{P}\{\mathbf{E}^D_{\rm tx}\}$. 
Unfortunately a general proof of this conjecture has eluded us. A proof under some additional conditions and $M=1$ was derived in \cite{Vishnu_Globecom}. While the rest of the results in the paper are exact for ${\rm rank}\{\mathbf{R}_{\rm tx}\} \leq D$, we shall rely on this intuitive, albeit difficult to prove, conjecture when ${\rm rank}\{\mathbf{R}_{\rm tx}\} > D$. 
From Theorem \ref{Th_restrict_dom_eigspace}, Conjecture \ref{conj_dom_eigspace} and \eqref{eqn_cap_MU_nearopt}-\eqref{eqn_channel_KL_expansion_full}, problem \eqref{eqn_prob_formulate} reduces to $\mathbf{T}_{\rm opt} = \mathbf{E}_{\rm tx}^D \hat{\mathbf{T}}_{\rm opt}$, where:
\begin{flalign}
& \hat{\mathbf{T}}_{\rm opt} = {\rm argmax}_{\hat{\mathbf{T}} \in \mathbb{C}^{D\times L}} \left \{ \underbar{C}^{D}(\hat{\mathbf{T}}) \right \}, \label{eqn_opt_prbm_D} & \\
& \underbar{C}^{D}(\hat{\mathbf{T}}) \triangleq \underbar{C}(\mathbf{E}_{\rm tx}^D  \hat{\mathbf{T}}) = \mathbb{E}_{\mathbf{H}^{D}} \bigg\{ \max_{1 \leq i \leq |\mathcal{S}|} \log \Big| \mathbb{I}_M + & \nonumber \\
&\quad \frac{\rho}{M} \mathbf{R}_{\rm rx}^{1/2} \mathbf{H}^{D} {[\boldsymbol{\Lambda}^D_{\rm tx}]}^{1/2} \hat{\mathbf{T}} \mathbf{S}_i \hat{\mathbf{G}}_i  \hat{\mathbf{G}}^{\dag}_i \mathbf{S}^{\dag}_i \hat{\mathbf{T}}^{\dag} {[\boldsymbol{\Lambda}^D_{\rm tx}]}^{1/2} {[\mathbf{H}^{D}]}^{\dag} \mathbf{R}_{\rm rx}^{1/2} \Big| \bigg\}, \!\!\!\!\!\!\! \label{eqn_def_fobjD} &
\end{flalign}
where $\boldsymbol{\Lambda}^D_{\rm tx}$ is the $D \times D$ principal submatrix of $\boldsymbol{\Lambda}_{\rm tx}$, $\mathbf{H}^D$ is a $M \times D$ matrix with i.i.d. $\mathcal{CN}(0, 1)$ entries and $\hat{\mathbf{G}}_i$ ortho-normalizes columns of $\hat{\mathbf{T}}\mathbf{S}_i$. 
Henceforth we shall restrict to finding the optimal solution $\hat{\mathbf{T}}_{\rm opt}$ to \eqref{eqn_opt_prbm_D}, since $\mathbf{T}_{\rm opt}$ can be found in a straightforward way from it. In fact expressing $\mathbf{T}_{\rm opt}$ as $\mathbf{E}_{\rm tx}^D \hat{\mathbf{T}}_{\rm opt}$ may also help reduce the hardware cost for implementing the analog beamforming matrix, as shall be shown later in Section \ref{sec_hardware_pp}. To prevent any confusion, we shall refer to $\hat{\mathbf{T}}_{\rm opt}$ as the reduced dimensional (RD) beamformer. Though \eqref{eqn_opt_prbm_D} reduces the search space from $\mathbb{C}^{N \times L}$ to $\mathbb{C}^{D \times L}$, it is still unbounded. This problem is remedied by the following theorem. 

\begin{theorem}[Bounding the search space] \label{Th_bound_eigspace}
For any $\hat{\mathbf{T}} \in \mathbb{C}^{D\times L}$, both $\hat{\mathbf{T}}$ and $\hat{\mathbf{T}} \boldsymbol{\Lambda}_{\theta}$ attain the same hSNR sum capacity \eqref{eqn_def_fobjD}, where $\boldsymbol{\Lambda}_{\theta}$ is any arbitrary $L \times L$ complex diagonal matrix. 
\end{theorem}
\begin{proof}
See Appendix \ref{appdix2}.
\end{proof}
From Theorem \ref{Th_bound_eigspace}, by replacing $\hat{\mathbf{T}}$ by $\hat{\mathbf{T}}_{\theta} = \hat{\mathbf{T}} \boldsymbol{\Lambda}_{\theta}$ in \eqref{eqn_opt_prbm_D}, where: 
\begin{eqnarray}
{[\Lambda_\theta]}_{\ell,\ell} = \frac{{[\hat{\mathbf{T}}_{1,\ell}]}^{\dag}}{ {\big|{[\hat{\mathbf{T}}]}_{{\rm c}\{\ell\}} \big|} |\hat{\mathbf{T}}_{1,\ell}|} \ \ \forall 1 \leq \ell \leq L, \nonumber
\end{eqnarray}
the optimal RD-beamformer design problem can be reduced to:
\begin{flalign}
& \hat{\mathbf{T}}_{\rm opt} = {\rm argmax}_{\hat{\mathbf{T}} \in \mathcal{T}_\mathcal{G}} \left \{ \underbar{C}^{D}(\hat{\mathbf{T}}) \right \} \ \ \text{where,} \label{eqn_equiv_prob_formulate2} & \\
& \ \mathcal{T}_{\mathcal{G}} = \left\{ \hat{\mathbf{T}} \in \mathbb{C}^{D\times L} \Big| \big|{[\hat{\mathbf{T}}]}_{{\rm c}\{\ell\}}\big| = 1, {\rm Im}\{\hat{\mathbf{T}}_{\ell,1}\} = 0 \ \forall \ell = 1,...,L \right\}, \nonumber &
\end{flalign}
where, ${\rm Im}\{\}$ represents the imaginary component. The constraints in \eqref{eqn_equiv_prob_formulate2} aid in resolving the ambiguities of $\hat{\mathbf{T}}$, similar to the approach used in \cite{Bjornson2014}. Note that since the hSNR sum capacity $\underbar{C}^{D}(\hat{\mathbf{T}})$ is invariant to complex scaling of the columns of $\hat{\mathbf{T}}$, each column ${[\hat{\mathbf{T}}]}_{{\rm c}\{\ell\}}$ for $1 \leq \ell \leq L$ is representative of $\mathcal{P}\{{[\hat{\mathbf{T}}]}_{{\rm c}\{\ell\}}\}$, i.e., it represents a point on the complex Grassmannian manifold $\mathcal{G}(D,1)$. For $a \geq b > 0$, the complex Grassmannian manifold $\mathcal{G}(a,b)$ is the set of all linear sub-spaces of dimension $b$ in $\mathbb{C}^{a \times 1}$. Therefore, \eqref{eqn_equiv_prob_formulate2} is actually an optimization problem over the complex Grassmannian manifold $\mathcal{G}(D,1)$.

\section{Lower bound on the objective function} \label{sec_approx_obj_fn}
Though transformations to the search space were introduced in the previous section to reduce the search complexity, the hSNR sum capacity $\underbar{C}^{D}(\hat{\mathbf{T}})$ is not in closed form. A closed-form lower bound to $\underbar{C}^{D}(\hat{\mathbf{T}})$ for the case of $M=1$ was considered in \cite{Vishnu_ICC2017} which was shown to be maximized by Grassmannian line packing the columns of $\hat{\mathbf{T}}$. However, this bound is independent of the switch position set $\mathcal{S}$ and cannot be generalized to $M > 1$. Similarly, another approximation to $\underbar{C}^{D}(\hat{\mathbf{T}})$ can be obtained via the work on restricted precoding \cite{Vishnu_jrnl1}. Though this approximation eliminates the need for taking an expectation as in \eqref{eqn_def_fobjD}, it has to be computed recursively and is accurate only when $D,K \gg M$. In contrast, in this section we find a closed-form lower bound to the sum capacity that depends on $\mathcal{S}$. Henceforth, for ease of analysis, we assume $\boldsymbol{\Lambda}_{\rm tx}^{D} = \mathbb{I}_D$.\footnote{Any constant scaling factor in $\boldsymbol{\Lambda}_{\rm tx}^{D}$ is included into $\rho$, without loss of generality.} Extension to more generic channels is considered later in Section \ref{sec_anisotropy}. 

For any $a \geq b > 0$, we define the complex Stiefel manifold $\mathcal{U}(a,b)$ as the set of all $a \times b$ matrices with ortho-normal columns. We shall refer to such matrices as semi-unitary matrices. For $\mathbf{A}, \mathbf{B} \in \mathcal{U}(a,b)$, we further define the `Fubini-Study distance' function as: 
\begin{eqnarray}
d_{\rm FS} \big(\mathbf{A}, \mathbf{B} \big) = \arccos \sqrt{|\mathbf{A}^{\dag} \mathbf{B} \mathbf{B}^{\dag} \mathbf{A}|}. \label{eqn_define_FS}
\end{eqnarray} 
Here $d_{\rm FS} (\mathbf{A}, \mathbf{B})$ is not a distance measure between $\mathbf{A}$ and $\mathbf{B}$, but rather a distance measure between $\mathcal{P}\{\mathbf{A}\}$ and $\mathcal{P}\{\mathbf{B}\}$ on $\mathcal{G}(a,b)$. 
For ease of notation, we further define $\mathbf{Q}_i \triangleq \hat{\mathbf{T}}\mathbf{S}_i \hat{\mathbf{G}}_i$ for each selection matrix $\mathbf{S}_i \in \mathcal{S}$. Note that $\mathbf{Q}_i \in\mathcal{U}(D,K)$ and $\mathcal{P}\{\mathbf{Q}_i\} \in \mathcal{G}(D,K)$ for all $i=1,..,|\mathcal{S}|$. 
We then have the following lemma:
\begin{lemma}[Higher dimension lower bound] \label{Th_subopt_selection}
If $\boldsymbol{\Lambda}^D_{\rm tx} = \mathbb{I}_D$, we have $\underbar{C}^{D}(\hat{\mathbf{T}}) \geq C^{D}_{\rm LB1}(\hat{\mathbf{T}})$ where: 
\begin{eqnarray}
C^{D}_{\rm LB1}(\hat{\mathbf{T}}) \triangleq \mathbb{E}_{\mathbf{H}^D} \mathbb{E}_{\mathbf{V}} \left\{ \max_{1 \leq i \leq |\mathcal{S}|} \log \left[ 1 + \alpha\left| \mathbf{V}^{\dag} \mathbf{Q}_i \mathbf{Q}^{\dag}_i \mathbf{V} \right| \right] \right\}, \label{eqn_th3_LB2}
\end{eqnarray}
$\alpha = {\left(\frac{\rho}{M} \right)}^M \big|\mathbf{R}_{\rm rx} \big| \big| \mathbf{H}^D {[\mathbf{H}^D]}^{\dag} \big|$, $\mathbf{H}^D$ is as defined in \eqref{eqn_def_fobjD} and $\mathbf{V}$ is a random matrix uniformly distributed over $\mathcal{U}(D,K)$, independent of $\mathbf{H}^D$.
\end{lemma}
\begin{proof}
See Appendix \ref{appdix3}
\end{proof}
Note that in \eqref{eqn_def_fobjD}, each $\mathbf{Q}_i$ is associated with a corresponding selection region: $\big\{ \mathbf{H}^D \in \mathbb{C}^{M\times D} \big| i  = {\rm argmax}_{1\leq j \leq |\mathcal{S}|} | \mathbb{I}_M + \frac{\rho}{M} \mathbf{R}_{\rm rx}^{1/2} \mathbf{H}^D \mathbf{Q}_j \mathbf{Q}^{\dag}_j {[\mathbf{H}^D]}^{\dag} \mathbf{R}_{\rm rx}^{1/2} | \big\}$. Essentially, Lemma \ref{Th_subopt_selection} finds a lower bound where these selection regions are changed to $\big\{ \mathbf{V} \in \mathcal{U}(D,K) \big| i  = {\rm argmax}_{1\leq j \leq |\mathcal{S}|} | \mathbf{V}^{\dag} \mathbf{Q}_j \mathbf{Q}^{\dag}_j \mathbf{V}| \big\}$. These regions are easier to bound than those in \eqref{eqn_def_fobjD}, as exploited by the following theorem. 
%
\begin{theorem}[Fubini-Study lower bound] \label{Th_grassmann_tighter}
If $\boldsymbol{\Lambda}^D_{\rm tx} = \mathbb{I}_D$ and $D \gg 1$, we have $\underbar{C}^{D}(\hat{\mathbf{T}}) \geq C^{D}_{\rm LB}(\hat{\mathbf{T}})$, where:
\begin{subequations}
\begin{flalign}
& C^{D}_{\rm LB}(\hat{\mathbf{T}}) \triangleq |\mathcal{S}| {\left( \frac{1-\cos^{2/K}(\delta/2)}{K} \right)}^{D K + \epsilon} \!\!\!\!\!\!\!\!\!\![\beta + \log \cos^2(\delta/2)], \label{eqn_grassmann_tighter1} & \\
& \qquad \ \delta = \min_{i \neq j} d_{\rm FS}(\mathbf{Q}_i,\mathbf{Q}_j) \triangleq f_{\rm FS}(\hat{\mathbf{T}}), & \\
& \qquad \ \beta = M \log \left(\frac{\rho}{M} \right) \!+\! \log\left| \mathbf{R}_{\rm rx} \right| + \!\! \sum_{m=1}^M \psi\left(D-m+1\right), \!\!\!\!\label{eqn_def_beta} &
\end{flalign}
\end{subequations}
$\psi()$ being the digamma function and $\epsilon = o(D)$ i.e., $\lim_{D \rightarrow \infty} \epsilon/D = 0$. Furthermore, if $\beta \geq 2$:
\begin{eqnarray}
{\rm argmax}_{\hat{\mathbf{T}} \in \mathcal{T}_{\mathcal{G}}} \left \{ C^{D}_{\rm LB}(\hat{\mathbf{T}}) \right \} \equiv {\rm argmax}_{\hat{\mathbf{T}} \in \mathcal{T}_{\mathcal{G}}} \left \{ f_{\rm FS}(\hat{\mathbf{T}}) \right \}. \label{eqn_grassmann_tighter2}
\end{eqnarray}
\end{theorem}
\begin{proof}
Let us define $\delta \triangleq \min_{i \neq j} d_{\rm FS}(\mathbf{Q}_i, \mathbf{Q}_j)$, and $\mathcal{Q}_i \subseteq \mathcal{G}(D,K)$ for $1 \leq i \leq |\mathcal{S}|$ as:
\begin{eqnarray}
\mathcal{Q}_i & \triangleq & \left \{ \mathcal{P}\{\mathbf{W}\} \Big| \mathbf{W} \in \mathcal{U}(D,K), d_{\rm FS}( \mathbf{W}, \mathbf{Q}_i) < \delta/2 \right \} \nonumber \\
&=& \left \{ \mathcal{P}\{\mathbf{W}\} \Big| \mathbf{W} \in \mathcal{U}(D,K), |\mathbf{W}^{\dag} \mathbf{Q}_i \mathbf{Q}_i^{\dag} \mathbf{W}| > \cos^2(\delta/2) \right \}. \nonumber
\end{eqnarray}
Now consider $\mathbf{V}$ uniformly distributed over $\mathcal{U}(D,K)$ as in Lemma \ref{Th_subopt_selection}. Since both $\mathbf{V}$ and $\mathbf{Q}_i$ are semi-unitary, we have $0 \leq |\mathbf{V}^{\dag} \mathbf{Q}_i \mathbf{Q}_i^{\dag} \mathbf{V}| \leq 1$ \cite{Roger1990}. By pessimistically assuming that $|\mathbf{V}^{\dag} \mathbf{Q}_i \mathbf{Q}_i^{\dag} \mathbf{V}| = 0$ when $\mathcal{P}\{\mathbf{V}\} \notin \bigcup_{i} \mathcal{Q}_i$ and $|\mathbf{V}^{\dag} \mathbf{Q}_i \mathbf{Q}_i^{\dag} \mathbf{V}| = \cos(\delta/2)$ when $\mathcal{P}\{\mathbf{V}\} \in \bigcup_i \mathcal{Q}_i$, we can lower bound $C^{D}_{\rm LB1}(\hat{\mathbf{T}})$ in \eqref{eqn_th3_LB2} as:
\begin{eqnarray}
C^{D}_{\rm LB1}(\hat{\mathbf{T}}) & \geq & \mathbb{P} \Big( \mathcal{P}\{\mathbf{V}\} \in \bigcup_i \mathcal{Q}_i \Big) \mathbb{E}_{\mathbf{H}^D} \big\{ \log [\alpha\cos^2(\delta/2)] \big\} \nonumber \\
&=& \mathbb{P} \Big( \mathcal{P}\{\mathbf{V}\} \in \bigcup_i \mathcal{Q}_i \Big) \left[ \beta + \log \cos^2(\delta/2) \right],  \label{eqn_FSnorm_2}
\end{eqnarray}
where $\beta$ is given by \eqref{eqn_def_beta} and follows from the results on log-determinant of a Wishart matrix \cite{Goodman1963}.\footnote{Note that $\big|\mathbf{H}^D {[\mathbf{H}^D]}^{\dag} \big|$ is the determinant of a $M \times M$ complex Wishart matrix with $D$ degrees of freedom.} 
Since $\mathbf{V}$ is uniformly distributed over $\mathcal{U}(D,K)$, based on results in \cite{Love_SMux, Barg_grassmann}, we have:
\begin{eqnarray}
\mathbb{P} \Big( \mathcal{P}\{\mathbf{V}\} \in \bigcup_i \mathcal{Q}_i \Big) \geq |\mathcal{S}| {\left( \frac{1-\cos^{2/K}(\delta/2)}{K} \right)}^{D K + \epsilon}, \label{eqn_packing_density}
\end{eqnarray}
where $\epsilon = o(D)$. Using \eqref{eqn_FSnorm_2}--\eqref{eqn_packing_density} and Lemma \ref{Th_subopt_selection}, we arrive at \eqref{eqn_grassmann_tighter1}. 

Note that $\hat{\mathbf{T}}$ affects $C^{D}_{\rm LB}(\hat{\mathbf{T}})$ only via the term $\delta$ (for a fixed $L$). Therefore, if the partial derivative of $C^{D}_{\rm LB}(\hat{\mathbf{T}})$ with respect to $\delta$ is non-negative, then maximizing $\delta$ is equivalent to maximizing $C^{D}_{\rm LB}(\hat{\mathbf{T}})$. The required condition can be found as $\frac{\partial C^{D}_{\rm LB}(\hat{\mathbf{T}})}{ \partial \cos^2(\delta/2) } \leq 0$ i.e.,
\begin{eqnarray}
(DK \!\!+\! \epsilon) \frac{\cos^{2/K}(\delta/2)}{K} [\beta \!+\! \log \cos^{2}(\delta/2)] \geq 1 \!-\! \cos^{2/K}(\delta/2). \label{eqn_123}
\end{eqnarray}
Since $ \cos(\delta) = \min_{i\neq j} \sqrt{|\mathbf{Q}_i^{\dag} \mathbf{Q}_j \mathbf{Q}_j^{\dag} \mathbf{Q}_i|} \geq 0$, we have $\cos^{2}(\delta/2) = \frac{\cos(\delta)+1}{2} \geq \frac{1}{2}$. 
Therefore a sufficient condition for \eqref{eqn_123} can be obtained as:
\begin{eqnarray}
(D + \frac{\epsilon}{K}) [\beta - \log 2] \geq 2^{1/K} -1. \nonumber 
\end{eqnarray}
Letting $|\epsilon| \leq K D/2$ for $D \gg 1$ \big(since $\epsilon = o(D)$\big), it can be verified that the above holds for $\beta \geq 2$. Thus, \eqref{eqn_grassmann_tighter2} follows. 
\end{proof}

Since the objective in \eqref{eqn_equiv_prob_formulate2} is not in closed form, for $\boldsymbol{\Lambda}_{\rm tx}^{D} = \mathbb{I}_D$, we consider the sub-optimal RD-beamformer design problem that maximizes  $f_{\rm FS}(\hat{\mathbf{T}})$ in \eqref{eqn_grassmann_tighter2}, i.e., we focus on finding: 
\begin{eqnarray}
\hat{\mathbf{T}}_{\rm FS} = {\rm argmax}_{\hat{\mathbf{T}} \in \mathcal{T}_{\mathcal{G}}} \left \{ \min_{i \neq j} d_{\rm FS} \big( \mathbf{Q}_i, \mathbf{Q}_j \big) \right \}. 
\label{eqn_opt_prbm_Grassmann}
\end{eqnarray}
While it only maximizes a lower bound $C^{D}_{\rm LB}(\hat{\mathbf{T}})$ to the sum capacity, the metric $f_{\rm FS}(\hat{\mathbf{T}})$ can be readily computed for each candidate $\hat{\mathbf{T}}$ unlike $\underbar{C}^{D}(\hat{\mathbf{T}})$ in \eqref{eqn_def_fobjD}. 

\subsection{Interpreting of the Fubini-Study distance metric - $f_{\rm FS}(\hat{\mathbf{T}})$}
Note that for $\boldsymbol{\Lambda}^D_{\rm tx} = \mathbb{I}_D$, the hSNR sum capacity of the RD-beamformer in \eqref{eqn_def_fobjD} can be alternately expressed as: 
\begin{eqnarray}
\underbar{C}^{D}(\hat{\mathbf{T}}) = \mathbb{E}_{\mathbf{H}^{D}} \max_{1\leq i \leq |\mathcal{S}|} \left\{ \underbar{C}^{D}_{i}(\hat{\mathbf{T}}, \mathbf{H}^{D}) \right\},
\end{eqnarray}
where, $\underbar{C}^{D}_{i}(\hat{\mathbf{T}}, \mathbf{H}^D) = \log \left| \mathbb{I}_M + \frac{\rho}{M} \mathbf{R}_{\rm rx}^{1/2} \mathbf{H}^{D} \mathbf{Q}_i  \mathbf{Q}_i^{\dag} {[\mathbf{H}^D]}^{\dag} \mathbf{R}_{\rm rx}^{1/2} \right|$. Based on the results on restricted precoding \cite{Vishnu_jrnl1}, these individual hSNR capacities $\underbar{C}^{D}_{i}(\hat{\mathbf{T}}, \mathbf{H}^{D})$ are approximately jointly Gaussian distributed with second order statistics given by:
\begingroup\makeatletter\def\f@size{9.7}\check@mathfonts
\begin{subequations}
\begin{flalign}
&\mathbb{E}\{ \underbar{C}^{D}_{i}(\hat{\mathbf{T}}, \mathbf{H}^D) \} \approx M\log \left( \frac{\rho K^{3/2}}{M \sqrt{K+1}}\right) + \log|\mathbf{R}_{\rm rx}|, \label{eqn_cap_mean} & \\
& \text{Crosscov}\!\! \left\{ \underbar{C}^{D}_{i}\!(\hat{\mathbf{T}}, \mathbf{H}^D), \underbar{C}^{D}_{j}\!(\hat{\mathbf{T}}, \mathbf{H}^D) \right\} \approx M \log \!\bigg[ \! 1 \!+\! \frac{{\|\mathbf{Q}_j^{\dag} \mathbf{Q}_i\|}^2_F}{K^2} \!\bigg] \!\!\!\!\!\!\! \label{eqn_cap_crosscov} & \\
& \qquad \qquad \qquad \quad \geq M \log \bigg[1 + \frac{{\cos \big(d_{\rm FS}(\mathbf{Q}_i, \mathbf{Q}_j)\big)}^{2/K}}{K}\bigg], \nonumber &
\end{flalign}
\end{subequations}
\endgroup
where the last step follows by applying the AM-GM inequality on eigenvalues of $\mathbf{Q}_j^{\dag} \mathbf{Q}_i \mathbf{Q}_i^{\dag} \mathbf{Q}_j$ and using \eqref{eqn_define_FS}. Therefore, by maximizing $f_{\rm FS}(\hat{\mathbf{T}})$ in \eqref{eqn_opt_prbm_Grassmann}, we minimize a lower bound to the largest cross-covariance term among the individual hSNR capacities $\{\underbar{C}^{D}_{i}(\hat{\mathbf{T}}, \mathbf{H}^D)\}$.\footnote{
If we replace the Fubini-Study distance $d_{\rm FS}(\cdot)$ in \eqref{eqn_opt_prbm_Grassmann} by chordal distance \cite{Love_SMux}, the corresponding solution $\hat{\mathbf{T}}_{\rm chord}$ \emph{exactly} minimizes the largest cross-covariance term in \eqref{eqn_cap_crosscov}.
However simulations show no improvement in sum capacity with this replacement. Hence we stick to the Fubini-Study distance.} 
This is an intuitively pleasing result, since reducing cross-covariance typically shifts the probability distribution of the maximum of a set of Gaussian random variables to the right \cite{Vitale}.

\section{Design of the RD-beamformer} \label{sec_pp_design}
Since \eqref{eqn_opt_prbm_Grassmann} tries to maximize the minimum Fubini Study distance between the subspaces $\{\mathcal{P}\{\mathbf{Q}_i\} | 1 \leq i \leq |\mathcal{S}|\}$, it may seem identical to the well studied problem of Grassmannian subspace packing, for which several efficient algorithms are available in literature (see \cite{Medra2014} and references therein). 
However there is a subtle difference, which stems from the fact that $\mathbf{Q}_i = \hat{\mathbf{T}} \mathbf{S}_i \hat{\mathbf{G}}_i $ for $i = \{1,..., |\mathcal{S}|\}$ are generated from the same RD-beamformer $\hat{\mathbf{T}}$. They are therefore coupled, making \eqref{eqn_opt_prbm_Grassmann} a \emph{coupled} Grassmannian sub-space packing problem. This is illustrated via a toy example in Fig.~\ref{Fig_illustrate_coupled_subspace}, where $\mathcal{P}\{\mathbf{Q}_i\}$'s are represented as planes passing through the origin. 
Here rotating $\mathcal{P}\{\mathbf{Q}_1\}$ about $\hat{\mathbf{T}}_{{\rm c}\{2\}}$ would require moving $\hat{\mathbf{T}}_{{\rm c}\{1\}}$, which may further change other sub-spaces that contain $\hat{\mathbf{T}}_{{\rm c}\{1\}}$, such as, $\mathcal{P}\{\mathbf{Q}_2\}, \mathcal{P}\{\mathbf{Q}_4\}$. Thus the $\mathcal{P}\{\mathbf{Q}_i\}$'s are coupled in general, and \eqref{eqn_opt_prbm_Grassmann} should rather be interpreted as trying to pack the columns ${[\hat{\mathbf{T}}]}_{{\rm c}\{\ell\}}$ such that the planes ($\mathcal{P}\{\mathbf{Q}_i\}$'s) are well separated.
\begin{figure}[!h]
\centering
\includegraphics[width=0.3\textwidth]{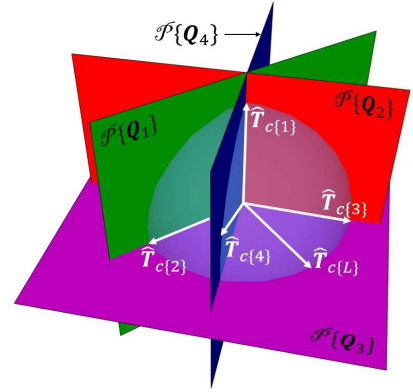}
\caption{An illustration of the coupled subspace packing problem \eqref{eqn_opt_prbm_Grassmann} for $D=3,K=2$ and a real field. Here $\hat{\mathbf{T}}\mathbf{S}_1 = \hat{\mathbf{T}}_{{\rm c}\{1,2\}}$, $\hat{\mathbf{T}} \mathbf{S}_2 = \hat{\mathbf{T}}_{{\rm c}\{1,3\}}$, $\hat{\mathbf{T}}\mathbf{S}_3 = \hat{\mathbf{T}}_{{\rm c}\{3,L\}}$ and $\hat{\mathbf{T}}\mathbf{S}_4 = \hat{\mathbf{T}}_{{\rm c}\{1,4\}}$.}
\label{Fig_illustrate_coupled_subspace}
\end{figure}
To the best of our knowledge, such a problem has not been addressed in literature before. 
\begin{lemma}[] \label{Th_opt_subspace_pack}
If $L \leq D$, any $\hat{\mathbf{T}} \in \mathcal{U}(D,L)$ is optimal for \eqref{eqn_opt_prbm_Grassmann}.
\end{lemma}
\begin{proof}
From \eqref{eqn_define_FS} it can be readily verified that $0 \leq d_{\rm FS}(\cdot) \leq \pi/2$. Therefore from \eqref{eqn_opt_prbm_Grassmann}, $0 \leq f_{\rm FS}(\hat{\mathbf{T}}) \leq \pi/2$. 

Now, for $L \leq D$, consider a $D\times L$ RD-beamformer $\hat{\mathbf{T}}$ such that $\hat{\mathbf{T}}^{\dag} \hat{\mathbf{T}} = \mathbb{I}_L$. For any $i,j \in \{1,..., |\mathcal{S}|\}$ and $i \neq j$, there exists $\ell \in \{1,.., L\}$ such that $\mathbf{S}_i$ picks $\hat{\mathbf{T}}_{{\rm c}\{\ell\}}$ but $\mathbf{S}_j$ does not. Then we have $\mathbf{Q}_j^{\dag} \hat{\mathbf{T}}_{{\rm c}\{\ell\}} = \mathbb{O}_{K \times 1}$, which follows from the fact that $\hat{\mathbf{T}}$ has orthonormal columns and hence $\mathcal{P}\{\hat{\mathbf{T}}\mathbf{S}_j \} \perp \hat{\mathbf{T}}_{{\rm c}\{\ell\}}$. Furthermore,
$\exists \mathbf{a} \in \mathbb{C}^{L \times 1}$ such that $\mathbf{Q}_i \mathbf{a} = \hat{\mathbf{T}}_{{\rm c}\{\ell\}}$. Then:
\begin{eqnarray}
\mathbf{a} \mathbf{Q}_i^{\dag} \mathbf{Q}_j \mathbf{Q}_j^{\dag} \mathbf{Q}_i \mathbf{a} &=& {[\hat{\mathbf{T}}_{{\rm c}\{\ell\}}]}^{\dag} \mathbf{Q}_j \mathbf{Q}_j^{\dag} \hat{\mathbf{T}}_{{\rm c}\{\ell\}} = 0 \nonumber \\
\Rightarrow d_{\rm FS} \big( \mathbf{Q}_i, \mathbf{Q}_j \big) &=& \pi/2.
\end{eqnarray}
Since $\hat{\mathbf{T}}$ satisfies the upper bound on $f_{\rm FS}(\cdot)$, the lemma follows. 
\end{proof}
Unfortunately, solutions to \eqref{eqn_opt_prbm_Grassmann} are not known for the more interesting case of $L > D$. However, a related problem is the problem of Grassmannian line packing:
\begin{eqnarray}
\hat{\mathbf{T}}_{\rm LP} = {\rm argmax}_{\hat{\mathbf{T}} \in \mathcal{T}_{\mathcal{G}}} \min_{i \neq j}\left \{ d_{\rm FS}({[\hat{\mathbf{T}}]}_{{\rm c}\{i\}}, {[\hat{\mathbf{T}}]}_{{\rm c}\{j\}}) \right \}, \label{eqn_opt_prbm_linepacking}
\end{eqnarray}
which tries to maximize the minimum Fubini Study distance between the columns of $\hat{\mathbf{T}}$, and for which several near-optimal solutions are available in literature \cite{Medra2014, Dhillon2008}. Both problems have identical solutions $\hat{\mathbf{T}}_{\rm FS} = \hat{\mathbf{T}}_{\rm LP}$ for $L \leq D$. 
While this is not true for $L > D$, we hypothesize that $\hat{\mathbf{T}}_{\rm LP}$ might still serve as a good, analytically tractable, sub-optimal solution to \eqref{eqn_opt_prbm_Grassmann}. One important difference however is that unlike $\hat{\mathbf{T}}_{\rm FS}$, $\hat{\mathbf{T}}_{\rm LP}$ is independent of the switch position set $\mathcal{S}$, and therefore may have poor performance for certain $\mathcal{S}$ if $L > D$. Therefore, we explore some numerical optimization algorithms to adapt $\hat{\mathbf{T}}_{\rm LP}$ to $f_{\rm FS}(\cdot)$ in Appendix \ref{appdix5}. These algorithms are used later in Section \ref{sec_sim_results} to evaluate the quality of the line packed solution $\hat{\mathbf{T}}_{\rm LP}$, via simulations. 

\section{Reducing the hardware and computational complexity} \label{sec_hardware_cost}
For each user group, we assumed in Section \ref{sec_chan_model} that $K,L$ are chosen by an inter-group resource sharing algorithm. Full flexibility in picking $K,L$ imposes a significant hardware cost for large values of $N$. Additionally, a large $|\mathcal{S}|$ may also increase the computational complexity of picking the best selection matrix for each channel realization. In this section we shall discuss methods to reduce these hardware and computational costs when $K,L$ are pre-fixed values i.e., inter-group resource sharing only involves power allocation. In this case, we can restrict discussion to the beamformer and selection bank for a single user group. 

\subsection{Reducing hardware cost of the beamforming matrix} \label{sec_hardware_pp}
In general, we need a variable gain phase-shifter for each element of the analog beamforming matrix $\mathbf{T}$, thereby, needing $NL$ components. This leads to a large implementation cost, especially when $L > D$. However if $\mathbf{T}$ is designed apriori for a fixed value of $D$, the hardware cost can be reduced significantly, as illustrated next. 
Note that the proposed beamforming matrix can be expressed as: $\mathbf{T} = \mathbf{T}_{\rm var} \hat{\mathbf{T}}_{\rm fix}$, where $\mathbf{T}_{\rm var} = \mathbf{E}_{\rm tx}^{D} {[\hat{\mathbf{T}}]}_{{\rm c}\{1:D\}}$, $\hat{\mathbf{T}}_{\rm fix} = \left[ \begin{array}{cc} \mathbb{I}_{D}  & {[\hat{\mathbf{T}}]}_{{\rm c}\{1:D\}}^{-1} {[\hat{\mathbf{T}}]}_{{\rm c}\{(D+1):L\}} \end{array} \right]$ and $\hat{\mathbf{T}}$ is a $D \times L$ RD-beamformer designed for either \eqref{eqn_opt_prbm_Grassmann} or \eqref{eqn_opt_prbm_linepacking}. Firstly, by implementing both components $\mathbf{T}_{\rm var}$ and $\hat{\mathbf{T}}_{\rm fix}$ separately as shown in Fig.~\ref{Fig_blockdiag_restrict_sel}, the number of required variable gain phase-shifters reduce to $D(N + L - D)$, which can be a significant reduction when $N \gg D$ and $L > D$. The number of analog power dividers may however increase by a factor of $D/K$.
Secondly, since design of $\hat{\mathbf{T}}$ is independent of aCSI given $D$ and $\boldsymbol{\Lambda}_{\rm tx}^{D} = \mathbb{I}_D$ (see \eqref{eqn_opt_prbm_Grassmann}), the $D(L-D)$ components of $\hat{\mathbf{T}}_{\rm fix}$ can be implemented using a fixed phase-shifter array. Later in Section \ref{sec_anisotropy} it shall be shown that this fixed structure is also applicable when $\boldsymbol{\Lambda}_{\rm tx}^{D} \neq \mathbb{I}_D$.
\begin{figure}[!h]
\centering
\includegraphics[width=0.48\textwidth]{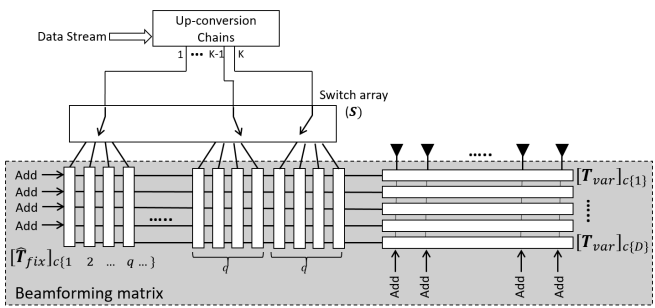}
\caption{Block diagram of a reduced complexity analog front-end design for HBwS, corresponding to one user group.}
\label{Fig_blockdiag_restrict_sel}
\end{figure}
Further reduction in the hardware complexity is possible using unit gain, discrete phase-shifter components for the beamformer \cite{Molisch_VarPhaseShift, Sudarshan, Alkhateeb2013}. The use and impact of such components, however, is beyond the scope of this paper. 
\subsection{Restricting the size of the switch position set} \label{sec_restrict_sel}
In this subsection, we restrict the size of the switch position set $\mathcal{S}$ for each user group. The size restriction not only reduces the hardware cost of implementing the switch bank, but also reduces the computational effort of picking the best selection matrix for a channel realization. In fact, since the $\mathbf{Q}_i$'s are coupled, some selection matrices may contribute little to the overall system performance.

Let us define for each selection matrix $\mathbf{S}_i$ a corresponding set $\mathcal{B}_i \subset \{1,..,L\}$ such that $\ell \in \mathcal{B}_i$ iff ${[\mathbf{S}_i]}_{{\rm c}\{k\}} = {[\mathbb{I}_L]}_{{\rm c}\{\ell\}}$ for some $1 \leq k \leq K$ i.e., $\mathbf{S}_i$ connects input port $\ell$ of $\mathbf{T}$ to some up-conversion chain. 
It can then be shown that $\mathbf{Q}^{\dag}_i \mathbf{Q}_j \mathbf{Q}^{\dag}_j \mathbf{Q}_i$ has $|\mathcal{B}_i \cap \mathcal{B}_j|$ unity eigenvalues (see Appendix \ref{appdix4}). 
Therefore, from the definition of $d_{\rm FS}(\cdot)$ in \eqref{eqn_define_FS}, we have for $i \neq j$:
\begin{subequations}\label{eqn_FS_UB}
\begin{flalign}
& {\cos \! \big(d_{\rm FS} (\mathbf{Q}_i, \mathbf{Q}_j) \big)}^2 \geq {\left[ \lambda^{\downarrow}_{K}\{ \mathbf{Q}^{\dag}_i \mathbf{Q}_j \mathbf{Q}^{\dag}_j \mathbf{Q}_i \}\right]}^{K-|\mathcal{B}_i \cap \mathcal{B}_j|} \!\!\!\!\!\!\!\!\!\!\!\!, & \\
& {\cos \! \big(d_{\rm FS} (\mathbf{Q}_i, \mathbf{Q}_j) \big)}^2 \leq {\left[ \lambda^{\downarrow}_{|\mathcal{B}_i \cap \mathcal{B}_j|+1}\{ \mathbf{Q}^{\dag}_i \mathbf{Q}_j \mathbf{Q}^{\dag}_j \mathbf{Q}_i \}\right]}^{K-|\mathcal{B}_i \cap \mathcal{B}_j|} \!\!\!\!\!\!\!\!\!\!\!\!, &
\end{flalign}
\end{subequations}
where $\lambda^{\downarrow}_k(\mathbf{A})$ represents the $k$-th largest eigenvalue of a matrix $\mathbf{A}$. Bounds in \eqref{eqn_FS_UB} suggest that reducing $|\mathcal{B}_i \cap \mathcal{B}_j|$ might help increase $\cos \! \big(d_{\rm FS} (\mathbf{Q}_i, \mathbf{Q}_j) \big)$. Therefore a good way of increasing $C^{D}_{\rm LB}(\hat{\mathbf{T}})$ in Theorem \ref{Th_grassmann_tighter} is to reduce $|\mathcal{B}_i \cap \mathcal{B}_j|$ for $i \neq j$. However, $|\mathcal{S}|$ should also be kept as large as possible to minimize the performance loss. 
In other words, we wish to find the largest family of subsets $\mathcal{\tilde{B}}$ such that:\footnote{\label{note3} $\mathcal{\tilde{B}} = \{ \mathcal{B}_1,..,\mathcal{B}_{|\tilde{\mathcal{B}}|} \}$ is a set of subsets of column indices of $\mathbf{T}$, each element of which corresponds to a selection matrix $\mathbf{S} \in \mathcal{S}$.} 
\begin{eqnarray}
\mathcal{\tilde{B}} &=& \big\{\mathcal{B}_1,..,\mathcal{B}_{|\mathcal{\tilde{B}}|} \big| \mathcal{B}_i \subseteq \{1,..,L\}, |\mathcal{B}_i| = K \nonumber \\
&& \qquad \qquad \text{ and } |\mathcal{B}_i \cap \mathcal{B}_j| \leq \kappa \ \forall i \neq j \big\}. \label{define_set_B_tilde}
\end{eqnarray}
Finding the largest such family is an open, but well studied, problem in the field of extremal combinatorics. Based on some of these results, we have the following theorem:
\begin{theorem}[$K$-uniform, $\{0:\kappa\}$-intersecting subsets] \label{Th_Kuniform_Mintersecting_sets}
Let $\tilde{\mathcal{B}}$ be the largest subset of the power set of $\{1,...,L\}$ such that \eqref{define_set_B_tilde} is satisfied. Then the cardinality of $\tilde{\mathcal{B}}$ satisfies: 
\begin{eqnarray} \label{eqn_overlap1_th}
{\left[\frac{L}{2 K}\right]}^{\kappa+1} \leq q^{\kappa+1} \leq |\mathcal{\tilde{B}}| \leq \frac{{}^{L} {\rm C}_{\kappa+1}}{{}^{K} {\rm C}_{\kappa+1}} & \text{if $L \geq 2 K^2$},
\end{eqnarray}
where $q$ is the largest prime number such that $q \leq L/K$.
\end{theorem}
\begin{proof}
The upper bound is derived in \cite{Deza1980} and an algorithm that achieves the lower bound was proposed in \cite[Theorem 4.11]{Babai1992}, which is reproduced below for convenience. 
Let $q$ be the largest prime number such that $q \leq L/K$. If $q \geq K$, a construction of a family of $q^{\kappa+1}$ subsets with the required, bounded overlap is given by Algorithm \ref{algo_frankl}. 
\begin{algorithm}
\caption{Frankl-Babai Construction \cite{Babai1992}}\label{algo_frankl}
\begin{algorithmic}[1]
\FOR{$i = 1$ to $q^{\kappa+1}$}
\STATE $\mathcal{B}_i = \phi$
\FOR{$j = 0$ to $\kappa$}
\STATE $a_j = {\rm mod}(i \ | \ q^{j}, q)$ \ \ 
\COMMENT {Here $|$ implies integer division}
\ENDFOR
\STATE $f(x) \triangleq \sum_{j=0}^{\kappa} a_j x^j$
\FOR{$k = 0$ to $K-1$}
\STATE $\mathcal{B}_i = \mathcal{B}_i \bigcup \big\{ k q + {\rm mod}(f(k),q)+1\big\}$
\ENDFOR
\ENDFOR
\end{algorithmic}
\end{algorithm}
Now from Bertrand's postulate \cite{Ramanujan1919, BertrandWolfram}, there always exists a prime number $q$ between $L/(2K)$ and $L/K$ i.e., $q \geq L/(2K)$. Therefore a sufficient condition for Algorithm \ref{algo_frankl} is: $L/(2K) \geq K$. This concludes the theorem.
\end{proof}
For $\mathcal{\tilde{B}}$ designed by Algorithm \ref{algo_frankl}, each subset $\mathcal{B}_i$ picks exactly one element in the interval $\big[ k q,(k+1)q \big)$ for $k=0,..,K-1$. Therefore the corresponding switch position set $\mathcal{S}_{\rm Alg1(\kappa)} = \{ \mathbf{S}_1,..,\mathbf{S}_{|\mathcal{S}_{\rm Alg1(\kappa)}|} \}$ can be implemented by equipping each up-conversion chain with a $1$-to-$q$ switch as depicted in Fig.~\ref{Fig_blockdiag_restrict_sel}, i.e., each up-conversion chain has an exclusive set of input ports to connect to. This leads to a significant saving in hardware cost as opposed a system with all possible selections. Note that this reduced complexity structure is analogous to the designs in \cite{Xu_iCSI, Rial2016, Gao2018} for conventional hybrid beamforming and antenna selection. 
\section{Channel Estimation Overhead} \label{sec_chan_est_overhead}
For performing HBwS to a user group, we require the knowledge of $\{\mathbf{R}_{\rm tx}, \widetilde{\mathbf{H}}_{1} \mathbf{T}, ..., \widetilde{\mathbf{H}}_{M_1} \mathbf{T}\}$ at the BS and $\{\widetilde{\mathbf{H}}_{m} \mathbf{T} \mathbf{S} \mathbf{G}\}$ at each user $m$. In this section we quantify the corresponding channel estimation overhead for a narrow-band orthogonal frequency division multiplexing (OFDM) system. While we had also assumed knowledge of $\{\mathbf{R}_{\rm rx,1}, ..., \mathbf{R}_{\rm rx,M_1}\}$ at BS in Section \ref{sec_chan_model}, this knowledge is not utilized in the proposed beamformer designs (see Sections \ref{sec_pp_design} and \ref{sec_anisotropy}). Note that since we assume different user groups have orthogonal channels, pilots can be reused across the user groups. Hence, without loss of generality, the channel estimation overhead is quantified by considering a single representative group. A study of the impact of channel quantization or estimation errors on system performance is beyond the scope of this paper. 

Several algorithms have been proposed to acquire aCSI statistics, such as $\mathbf{R}_{\rm tx}$, with minimal training \cite{Caire2017, Park2016}. Additionally, since $\mathbf{R}_{\rm tx}$ remains constant for a long time duration and over a large bandwidth \cite{ispas2015analysis, wang2017stationarity}, it can be acquired at the BS with low overhead via uplink channel training, in both time division duplexing (TDD) and frequency division duplexing (FDD) systems.
Similarly, the acquisition of $\widetilde{\mathbf{H}}_{m} \mathbf{T} \mathbf{S} \mathbf{G}$ at each user $m$ imposes a small overhead, since each element of $\mathbf{u}$ in \eqref{eqn_downlink_y2} can use a different pilot sub-carrier and all the users can be trained in parallel via downlink training, once $\mathbf{S} \in \mathcal{S}$ is picked. 
The main bottleneck is the estimation of $\{\widetilde{\mathbf{H}}_{1} \mathbf{T}, ..., \widetilde{\mathbf{H}}_{M_1} \mathbf{T}\}$ at the BS. Note that it is sufficient to estimate $\{\widetilde{\mathbf{H}}_{1} \ddot{\mathbf{T}}, ..., \widetilde{\mathbf{H}}_{M_1} \ddot{\mathbf{T}}\}$, where $\ddot{\mathbf{T}}$ is a sub-matrix of $\mathbf{T}$ whose columns form a basis for $\mathcal{P}\{\mathbf{T}\}$. Since $|\mathcal{P}\{\mathbf{T}\}| \leq \min\{D,L\}$ from \eqref{eqn_prob_formulate}, this involves estimation of $\min\{D,L\} M$ channel coefficients. These coefficients can be obtained either via uplink pilot training in TDD, or via downlink training and feedback of $\widetilde{\mathbf{H}}_{m} \ddot{\mathbf{T}}$ from each user $m$ in FDD. 
As an illustration, we consider a TDD based system where the $M_1$ users transmit $\lceil \min\{D,L\}/K \rceil$ uplink pilot symbols in each coherence time. All the $M = M_1 M_2$ user antennas use orthogonal pilot sub-carriers for parallel training. By using a sequence of $\mathbf{S} \in \mathcal{S}$ for the $\lceil \min\{D,L\}/K \rceil$ pilots, such that each column of $\ddot{\mathbf{T}}$ is picked at least once, $\{\widetilde{\mathbf{H}}_{1} \ddot{\mathbf{T}}, ..., \widetilde{\mathbf{H}}_{M_1} \ddot{\mathbf{T}}\}$ can be estimated at the BS. The corresponding system sum-throughput (including channel estimation overhead) can be expressed as $O_{\rm HBwS}\underbar{C}(\mathbf{T})$, where: 
\begin{eqnarray} \label{eqn_chan_est_overhead}
O_{\rm HBwS} = 1 - \lceil \min\{D,L\}/K \rceil \zeta,
\end{eqnarray}
and $\zeta=(\text{symbol duration})\big/ (\text{coherence time})$. 
As is evident, there is a trade-off between the hSNR sum-capacity $\underbar{C}(\mathbf{T})$, which is an non-decreasing function of $D$ (see \eqref{eqn_prob_formulate}), and estimation overhead $O_{\rm HBwS}$, which is a non-increasing function of $D$. As mentioned in Section \ref{sec_prbm_formulate}, a good beamformer that maximizes the system throughput $O_{\rm HBwS}\underbar{C}(\mathbf{T})$ can therefore be obtained by performing a line search over $D \in \{1,...,N\}$. However proposing a computationally efficient algorithm to find this $D^{*}$ is beyond the scope of this paper (see \cite{Vishnu_Globecom, Gao2018} for some investigations). 
\section{Simulation Results} \label{sec_sim_results}
For simulations we consider a TDD based narrow-band OFDM system, with one BS ($N=100$) implementing HBwS and one representative user group. We assume the shared spatial TX correlation matrix has isotropic scattering within the dominant $D$-dimensional subspace, i.e., $\boldsymbol{\Lambda}_{\rm tx}^{D} = \mathbb{I}_D$ and $\mathbf{E}_{\rm tx}$ may be arbitrary. Extensions to the anisotropic case are considered in the next section. 
The switch bank provides each up-conversion chain with an exclusive set of $\lfloor L/K \rfloor$ input ports for connection \cite{Gao2017}. Unless otherwise stated, we assume that all the switch positions possible with this architecture are allowed i.e.,: 
\begin{eqnarray} 
\mathcal{S}_{\rm all} = \left\{ {[\mathbb{I}_L]}_{{\rm c}\{\ell_1,..,\ell_K\}} \Big| (k\!-\! 1) \left\lfloor \frac{L}{K} \right\rfloor \! < \! \ell_k \! \leq k \! \left\lfloor \frac{L}{K} \right\rfloor, k \in \{1,..,K\} \right\}. \nonumber 
\end{eqnarray}
For the results we use the system sum throughput $O_{\rm HBwS}\underbar{C}^{D}(\hat{\mathbf{T}})$ as the metric, where $O_{\rm HBwS}$ is from \eqref{eqn_chan_est_overhead}. Since $\underbar{C}^{D}(\hat{\mathbf{T}})$ in \eqref{eqn_def_fobjD} is not known in closed form, throughout this section we use Monte-Carlo simulations to obtain its sample-mean estimate. For each channel realization, a brute-force search is performed to pick the best $\mathbf{S} \in \mathcal{S}$. The design of low-complexity algorithms for selecting $\mathbf{S}$ is beyond the scope of this paper (see \cite{Molisch_mag,Magazine_nosratinia} and references therein). 
Note that the system model and capacity bound in Sections \ref{sec_chan_model}--\ref{sec_prbm_formulate} are also applicable to HBaCSI and HBiCSI, by setting $L=K$, $\mathcal{S}=\{\mathbb{I}_K\}$ (no selection stage) and letting $\mathbf{T}$ depend on aCSI and iCSI, respectively. For limiting the channel estimation overhead, we also restrict the beamforming for HBaCSI and HBiCSI to lie in the dominant $D$-dimensional channel subspace. 
Their hSNR sum capacity, with hSNR sum capacity maximizing beamformers, can be obtained by replacing $\mathcal{S}=\{\mathbb{I}_K\}$ with $\hat{\mathbf{T}}_{\text{HBaCSI}} = {[\mathbb{I}_D]}_{{\rm c}\{1:K\}}$ and $\hat{\mathbf{T}}_{\text{HBiCSI}} = \mathbf{E}^{K}_{\text{iCSI}}$ in \eqref{eqn_def_fobjD}, respectively, where $\mathbf{E}^{K}_{\text{iCSI}}$ is the $D\times K$ eigen-vector matrix of ${[\mathbf{E}_{\rm tx}^{D}]}^{\dag} \widetilde{\mathbf{H}}^{\dag} \widetilde{\mathbf{H}} \mathbf{E}_{\rm tx}^{D}$, corresponding to the $K$ largest eigenvalues. 
Similarly, their throughput can be obtained by using pre-log factors of $O_{\rm HBaCSI} = 1 - \zeta$ and $O_{\rm HBiCSI} = 1 - \lceil D/K \rceil\zeta$ in \eqref{eqn_def_fobjD}, respectively. A proof of the above is skipped for brevity. The choice of $D$ in the following results is arbitrary, and therefore further gains with HBwS are expected if a line search over $D$ is performed.

\subsection{Influence of number of input ports ($L$)} \label{subsec_influence_of_L}
The system sum throughput with HBwS as a function of number of input ports ($L$) is studied in Fig.~\ref{fig_influence_of_L}. Here we plot the performance of both the line packed RD-beamformer - $\hat{\mathbf{T}}_{\rm LP}$ and the RD-beamformer output of Algorithm \ref{algo_grad_descent}  - $\hat{\mathbf{T}}_{\rm Alg3}$ (see Appendix \ref{appdix5}). 
\begin{figure}[!h]
\centering
\includegraphics[width= 0.45\textwidth]{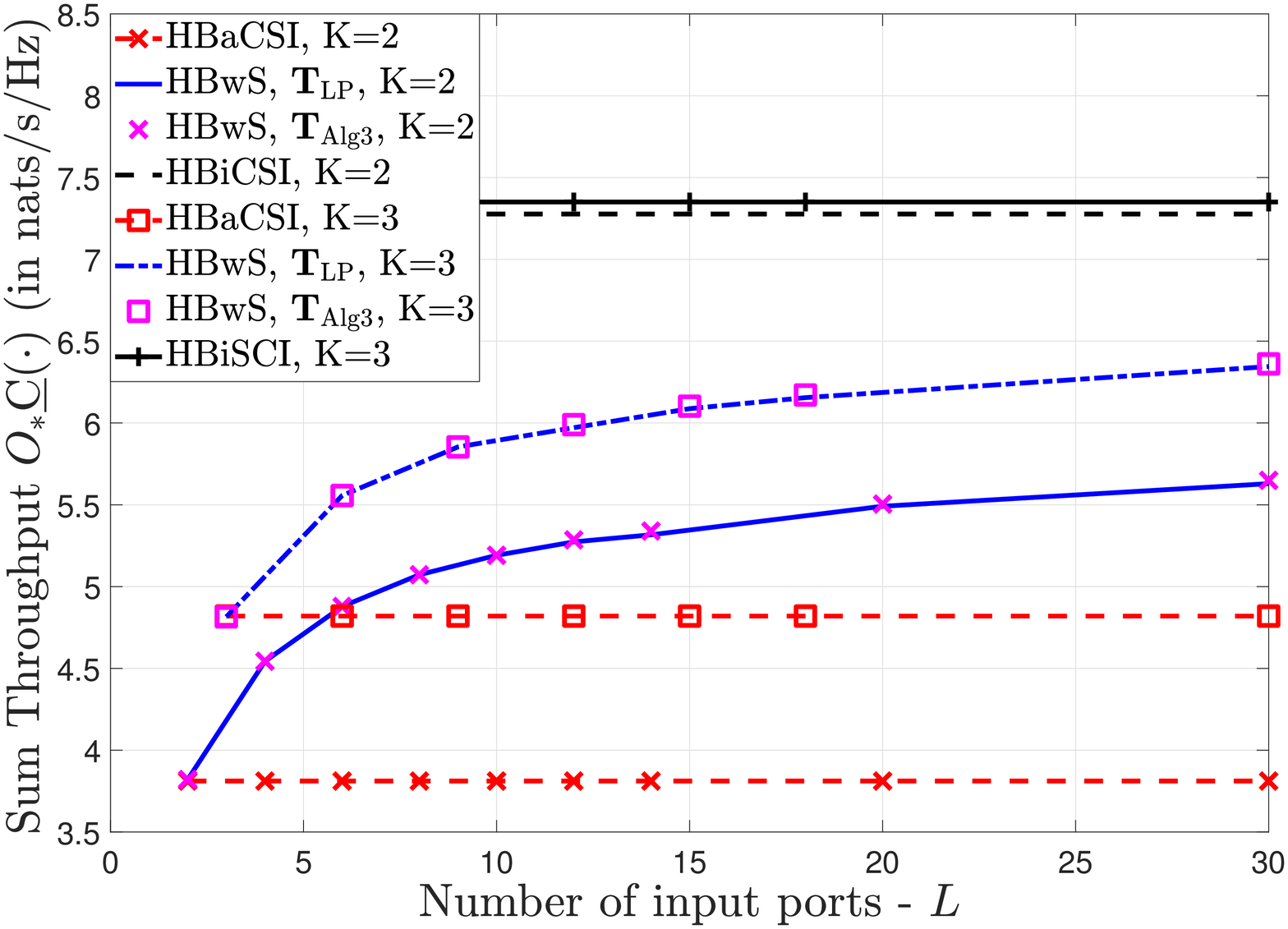}
\caption{Comparison of sum throughputs for different beamformers, as a function of $L$ \big($N=100,D=10,M=2,\rho=10$, $\mathcal{S} = \mathcal{S}_{\rm all}$, $\hat{\mathbf{T}}_{\rm LP}$ is generated using the line packing algorithm in \cite{MedraRepository}, $\boldsymbol{\Lambda}^{D}_{\rm tx} = \mathbb{I}_D$, $\mathbf{R}_{\rm rx} = \mathbb{I}_M$, $\zeta = 10^{-2}$\big).}
\label{fig_influence_of_L}
\end{figure}
The results suggest that with $L = 2D$, HBwS outperforms HBaCSI by approximately half the throughput gap between HBiCSI and HBaCSI. We also observe that there is a diminishing increase in throughput as we increase the number of ports $L$. Therefore for a good trade-off between hardware cost and performance, $L$ should be of the order of $D$. Finally, we observe that $\hat{\mathbf{T}}_{\rm LP}$ and $\hat{\mathbf{T}}_{\rm Alg3}$ have almost identical throughput for $\mathcal{S} = \mathcal{S}_{\rm all}$. 

\subsection{Influence of restricted switch positions ($\mathcal{S}$)}
In this sub-section, we analyze the performance of HBwS under the restricted switch position sets $\mathcal{S}_{\rm Alg1(\kappa)}$ from Algorithm \ref{algo_frankl}. Note that $\mathcal{S}_{\rm Alg1(\kappa)} \subseteq \mathcal{S}_{\rm all}$ and can therefore be implemented using the switch bank considered here. The sum-throughput $\underbar{C}^D(\hat{\mathbf{T}}_{\rm LP})$ with $\mathcal{S}_{\rm Alg1(\kappa)}$ is compared to the average sum-throughput with $\hat{\mathbf{T}}_{\rm LP}$ and a randomly chosen switch position set $\mathcal{S}_{\rm rand(\kappa)}$ of same size as $\mathcal{S}_{\rm Alg1(\kappa)}$, averaged over several realizations, in Fig.~\ref{fig_compare_overlapM}. The results support our claim that selection matrices with low overlap contribute more to hSNR sum capacity than others. Results also show that while $|\mathcal{S}_{\rm Alg1(\kappa)}|$ increases exponentially with $\kappa$, the increase in throughput is sub-linear, and therefore a small value of $\kappa$ is sufficient to achieve good performance. 
\begin{figure}[!h]
\centering
\includegraphics[width= 0.45\textwidth]{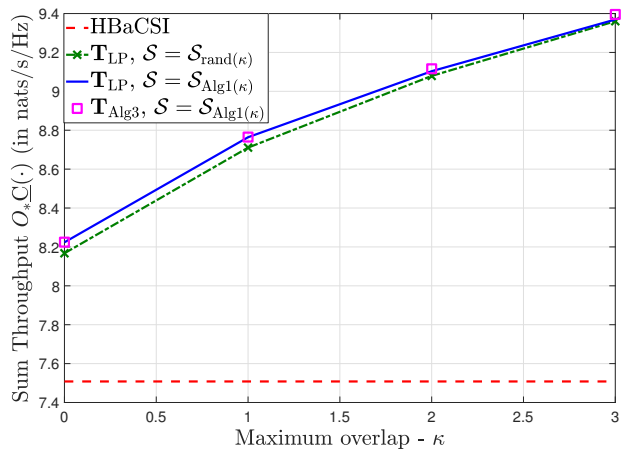}
\caption{Comparison of $O_{\rm HBwS}\underbar{C}^D(\hat{\mathbf{T}}_{\rm LP})$ with $\mathcal{S} = \mathcal{S}_{\rm Alg1(\kappa)}$ to the average sum-throughput $\mathbb{E}\{O_{\rm HBwS}\underbar{C}^D(\hat{\mathbf{T}}_{\rm LP})\}$ with $\mathcal{S} = \mathcal{S}_{\rm rand(\kappa)}$, averaged over several realizations of $\mathcal{S}_{\rm rand(\kappa)}$, as a function of maximum overlap $\kappa$. Here, $\mathcal{S}_{{\rm Alg1}(\kappa)}$ is designed using Algorithm \ref{algo_frankl} and $\mathcal{S}_{\rm rand(\kappa)}$ is a random subset of $\mathcal{S}_{\rm all}$ such that $|\mathcal{S}_{\rm rand(\kappa)}| = |\mathcal{S}_{\rm Alg1(\kappa)}|$. For $\mathcal{S} = \mathcal{S}_{\rm Alg1(\kappa)}$, we also plot $O_{\rm HBwS}\underbar{C}^D(\hat{\mathbf{T}}_{\rm Alg3})$ \big($N=100,D=10,L=20,K=M=4,\rho=10$, $\hat{\mathbf{T}}_{\rm LP}$ is from \cite{MedraRepository}, $\boldsymbol{\Lambda}^{D}_{\rm tx} = \mathbb{I}_D$, $\mathbf{R}_{\rm rx} = \mathbb{I}_M$, $\zeta = 10^{-2}$\big).}
\label{fig_compare_overlapM}
\end{figure}
Fig.~\ref{fig_compare_overlapM} also compares the sum throughput for $\hat{\mathbf{T}}_{\rm LP}$ and $\hat{\mathbf{T}}_{\rm Alg3}$, with $\mathcal{S}=\mathcal{S}_{\rm Alg1(\kappa)}$. The results suggest that even for these reduced complexity switch position sets, the line packing solution $\hat{\mathbf{T}}_{\rm LP}$ is almost as good as $\hat{\mathbf{T}}_{\rm Alg3}$. A similar trend has also been observed for several other switch position sets not discussed in this paper for brevity. This suggests that the design $\hat{\mathbf{T}}_{\rm LP}$ cannot be improved upon via conventional approaches such as Algorithms \ref{algo_permute_mtx} \& \ref{algo_grad_descent}, atleast for the $\mathcal{S}$ considered here.
The good performance of $\hat{\mathbf{T}}_{\rm LP}$ is probably because $f_{\rm FS}(\hat{\mathbf{T}}_{\rm LP}) \approx \pi/2$ when $L$ is of the order of $D$, and hence is near optimal for \eqref{eqn_opt_prbm_Grassmann} (see Theorem \ref{Th_opt_subspace_pack}). Henceforth, we shall avoid the use of the computationally intensive Algorithms \ref{algo_permute_mtx} \& \ref{algo_grad_descent}, and only restrict to use of $\hat{\mathbf{T}}_{\rm LP}$ to study performance of HBwS.

\subsection{Influence of number of users ($M_1$)}
We next consider the case where the representative user group has $M_1$ single antenna users i.e., $M_2=1$. The sum throughput of HBwS (normalized by sum throughput of HBiCSI) for varying $M_1$ is studied in Fig.~\ref{Fig_influence_L_MU}. Such a normalization allows comparison across different values of $M_1$. From the results, we observe that the slope of the normalized throughput curves increase with $M_1$, suggesting that the additional beam choices with HBwS aid in user separability. Apart from $O_{\rm HBwS}\underbar{C}^D(\hat{\mathbf{T}}_{\rm LP})$, which can be achieved via Dirty Paper Coding (DPC) \cite{Caire_DPC}, we also plot in Fig.~\ref{Fig_influence_L_MU} the normalized sum-throughput for Zero-forcing (ZF) precoding. Note that, unlike DPC, ZF precoding may yield poor performance when simultaneously supporting all $M_1$ users \cite{Dimic2005}. Therefore we allow scheduling of a sub-set of users ($\mathcal{M}_{\text{sc}}$) from the user group. The corresponding sum-rate with max-min fairness to the scheduled users can be expressed as:
\begin{flalign} 
& C^{D}_{\rm ZF}(\hat{\mathbf{T}}) = \mathbb{E}_{\widetilde{\mathbf{H}}} \bigg\{ \max_{1 \leq i \leq |\mathcal{S}|} \bigg[ |\mathcal{M}_{\text{sc}}| & \nonumber \\
& \qquad \qquad \log \bigg(1 \!+\! \frac{\rho}{\text{Tr} \big\{ {\big( \widetilde{\mathbf{H}}_{\text{sc}} \mathbf{E}_{\rm tx}^{D} \hat{\mathbf{T}} \mathbf{S}_i \hat{\mathbf{G}}_i \hat{\mathbf{G}}^{\dag}_i \mathbf{S}_i^{\dag} \hat{\mathbf{T}}^{\dag} {[\mathbf{E}_{\rm tx}^{D}]}^{\dag} \widetilde{\mathbf{H}}^{\dag}_{\text{sc}} \big) }^{-1} \big\}} \bigg) \bigg] \bigg\}, \!\!\!\!\!\!\!\!\! \nonumber &
\end{flalign}
where $\widetilde{\mathbf{H}}_{\text{sc}}$ is a sub-matrix of $\widetilde{\mathbf{H}}$ (see \eqref{eqn_cap_MU_nearopt}) corresponding to the scheduled users $\mathcal{M}_{\text{sc}}$. 
%
\begin{figure}[!h]
\centering
\includegraphics[width= 0.45\textwidth]{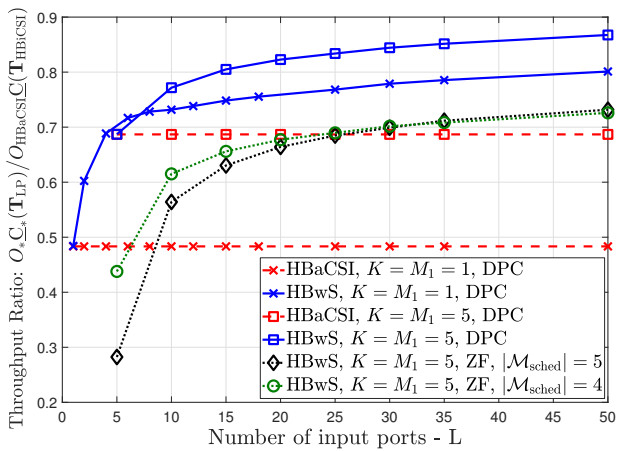}
\caption{Sum throughput of HBwS (normalized by the sum throughput of HBiCSI) versus $L$, for varying number of users. For ZF precoding, the first $|\mathcal{M}_{\text{sc}}|$ users in the user group are $\text{scheduled}^{9}$ \big($N=100, D=10, M_2 = 1, \rho=10$, $\mathcal{S} = \mathcal{S}_{\rm all}$, $\hat{\mathbf{T}}_{\rm LP}$ is from \cite{MedraRepository}, $\boldsymbol{\Lambda}^{D}_{\rm tx} = \mathbb{I}_D$, $\mathbf{R}_{\rm rx} = \mathbb{I}_{M_{1}}$, $\zeta = 10^{-2}$\big).}
\label{Fig_influence_L_MU}
\end{figure}
\footnotetext{Sophisticated user scheduling algorithms such as \cite{Dimic2005} cannot be directly extended to HBwS due to the presence of switching.}
Results show that HBwS helps reduce the throughput gap between the linear precoding scheme (ZF) and DPC, without requiring sophisticated user scheduling. 

\subsection{Influence of the bound on $\big| \mathcal{P}\{\mathbf{T}\} \big|$ - $D$}
As mentioned in Section \ref{sec_chan_est_overhead}, there exists a trade-off between hSNR sum capacity and estimation overhead as the value of $D$ changes. To illustrate this trade-off, we study the sum throughput of HBwS, HBaCSI and HBiCSI as a function of $D$ in Fig.~\ref{Fig_influence_D}, for a channel with isotropic scattering, i.e., $\mathbf{R}_{\rm tx} = \mathbb{I}_N$. 
\begin{figure}[!h]
\centering
\includegraphics[width= 0.45\textwidth]{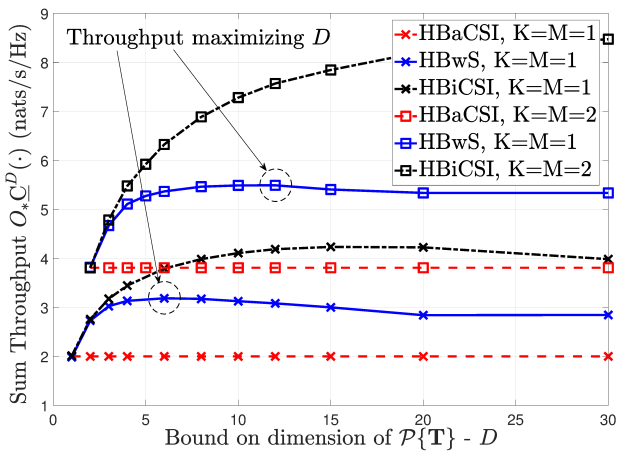}
\caption{Sum throuhgput of different schemes, as a function of $D$ \big($N=100, L=20, \rho=10$, $\mathcal{S} = \mathcal{S}_{\rm all}$, $\hat{\mathbf{T}}_{\rm LP}$ is from \cite{MedraRepository}, $\mathbf{R}_{\rm tx} = \mathbb{I}_N$, $\mathbf{R}_{\rm rx} = \mathbb{I}_M$, $\zeta= 10^{-2}$ \big).}
\label{Fig_influence_D}
\end{figure}
As evident from the results, there exists a $D^{*}$ that maximizes the sum throughput. Furthermore, this $D^{*}$ increases with $K$. Proposing a computationally efficient algorithm to find $D^{*}$ is beyond the scope of this paper. However a simple rule of thumb is to ensure that $D \ll K/\zeta$ to reduce overhead. 

\section{Anisotropic Channels} \label{sec_anisotropy}
In this section, we extend the beamformer design in Sections \ref{sec_restrict_search_space}-\ref{sec_pp_design} to anisotropic scattering in the dominant $D$ dimensional sub-space, i.e., $\boldsymbol{\Lambda}_{\rm tx}^{D} \neq \mathbb{I}_D$. Without loss of generality, we assume that the eigenvalues in $\boldsymbol{\Lambda}_{\rm tx}^{D}$ are arranged in the descending order of magnitude. Similar to the approach used in \cite{Love_beamforming_correlated, Xia2006, Raghavan_RVQ}, we employ the companding trick to adapt the beamforming matrix to $\boldsymbol{\Lambda}_{\rm tx}^{D}$, as:
\begin{eqnarray}
\mathbf{T}_{\rm ani} = \mathbf{E}^D_{\rm tx} {\boldsymbol{\Lambda}^{D}_{\rm tx}} \hat{\mathbf{T}}_{\rm LP}, \label{eqn_skewed_preprocessor}
\end{eqnarray}
where $\hat{\mathbf{T}}_{\rm LP}$ is the $D \times L$ line packed RD-beamformer. Intuitively, $\boldsymbol{\Lambda}^{D}_{\rm tx}$ in \eqref{eqn_skewed_preprocessor} skews the columns of $\mathbf{T}_{\rm ani}$, and therefore also the precoding beams $\mathbf{T}_{\rm ani} \mathbf{S}_i \mathbf{G}_i$, to be more densely packed near the eigen-vectors corresponding to the larger eigenvalues of $\mathbf{R}_{\rm tx}$.\footnote{A more general approach involves using ${[\boldsymbol{\Lambda}^{D}_{\rm tx}]}^{\xi}$ in \eqref{eqn_skewed_preprocessor}, where $\xi > 0$ is a skewing parameter.} 
For improving throughput, an additional optimization of $\mathbf{T}_{\rm ani}$ over the choice of $D$ is required (see Section \ref{sec_chan_est_overhead}). 
Note that this skewed beamformer can still be implemented using the two stage design in Sec \ref{sec_hardware_pp} by using $\mathbf{T}_{\rm var} = \mathbf{E}_{\rm tx}^{D} \boldsymbol{\Lambda}_{\rm tx}^{D} {[\hat{\mathbf{T}}_{\rm LP}]}_{{\rm c}\{1:D\}}$. 
For $L \leq D$, not all line packed matrices $\hat{\mathbf{T}}_{\rm LP}$ yield good performance after skewing by \eqref{eqn_skewed_preprocessor}. Therefore for $L \leq D$, we suggest the use of $\hat{\mathbf{T}}_{\rm LP} = { \left[ \begin{array}{cc} \mathbf{W} & \mathbb{O}_{(D-L)\times L} \end{array} \right]}^{\dag}$ in \eqref{eqn_skewed_preprocessor}, where $\mathbf{W}$ is the $L\times L$ DFT matrix, i.e., ${[\mathbf{W}]}_{a,b} = e^{\frac{j2\pi a b}{L}}/\sqrt{L}$. 

For simulations, we assume the BS has a half wavelength ($\lambda/2$) spaced uniform planar array of dimension $40 \times 10$. The TX transmits to a user-group of $M_1=3$ single antenna receivers that share the same transmit power angle spectrum (PAS), given by: 
\begingroup\makeatletter\def\f@size{9.5}\check@mathfonts
\begin{flalign}
\text{PAS}(\theta,\phi) = \sum_{i=1}^3 \exp{\left[-\eta |\theta \!-\! \tilde{\theta}_i | \!-\! \eta |\phi \!-\! \tilde{\phi}_i | \right]} \Pi(\theta \!-\! \tilde{\theta}_i,\phi \!-\! \tilde{\phi}_i) \label{eqn_PAS} \\
\text{where: } \Pi(\theta,\phi) = \left\{\begin{array}{cc} 1 & \text{for } |\theta | \leq \pi/20, |\phi| \leq \pi/20  \\ 
0 & \text{otherwise} \end{array} \right. , \nonumber
\end{flalign}
\endgroup
$\boldsymbol{\tilde{\theta}} = \left[ \frac{-3\pi}{10},0,\frac{\pi}{5} \right]$, $\boldsymbol{\tilde{\phi}} = \left[ \frac{6\pi}{10},\frac{8\pi}{10},\frac{7\pi}{10} \right]$, $\theta \in [-\pi/2, \pi/2)$ is the azimuth angle of departure, $\phi \in [0, \pi)$ is the elevation angle of departure and $\eta$ is a factor that controls the anisotropy of the channel. For this PAS, the transmit correlation matrix can be computed as in \eqref{eqn_compute_Rcorr}, where $d_H(a,b), d_V(a,b)$ are the horizontal and vertical spacing (in wavelengths) between elements $a$ and $b$, respectively. 
\begin{figure*}[!h]
\begin{eqnarray}
{[\mathbf{R}_{\rm tx}]}_{ab} = \frac{\int_{-\pi/2}^{\pi/2} \int_{0}^{\pi} \text{PAS}(\theta,\phi) e^{j 2\pi d_H(a,b) \sin(\phi) \sin(\theta) + j 2\pi d_V(a,b) \cos(\phi)} \sin(\phi) d\phi d\theta}{\int_{-\pi/2}^{\pi/2} \int_{0}^{\pi} \text{PAS}(\theta,\phi) \sin(\phi) d\phi d\theta}. \label{eqn_compute_Rcorr}
\end{eqnarray}
\end{figure*}
\begin{figure}[!h]
\centering
\subfloat[$\{D,L,K,M\} = \{24,9,3,3\}$]{\label{Fig_compare_anisotropic_1}\includegraphics[width= 0.43\textwidth]{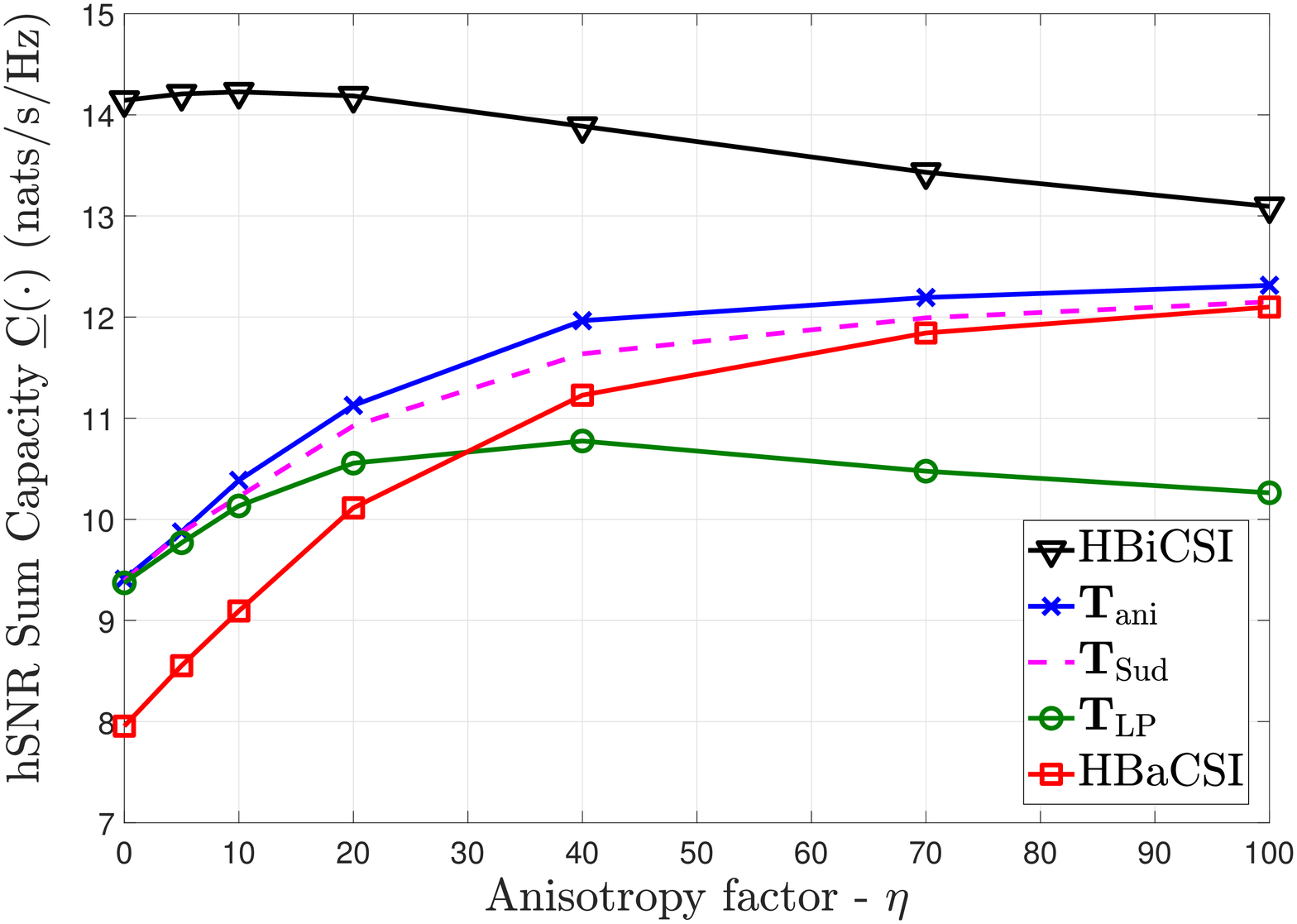}}
\hspace{1 mm}
\subfloat[$\{D,L,K,M\} = \{24,51,3,3\}$]{\label{Fig_compare_anisotropic_2}\includegraphics[width= 0.43\textwidth]{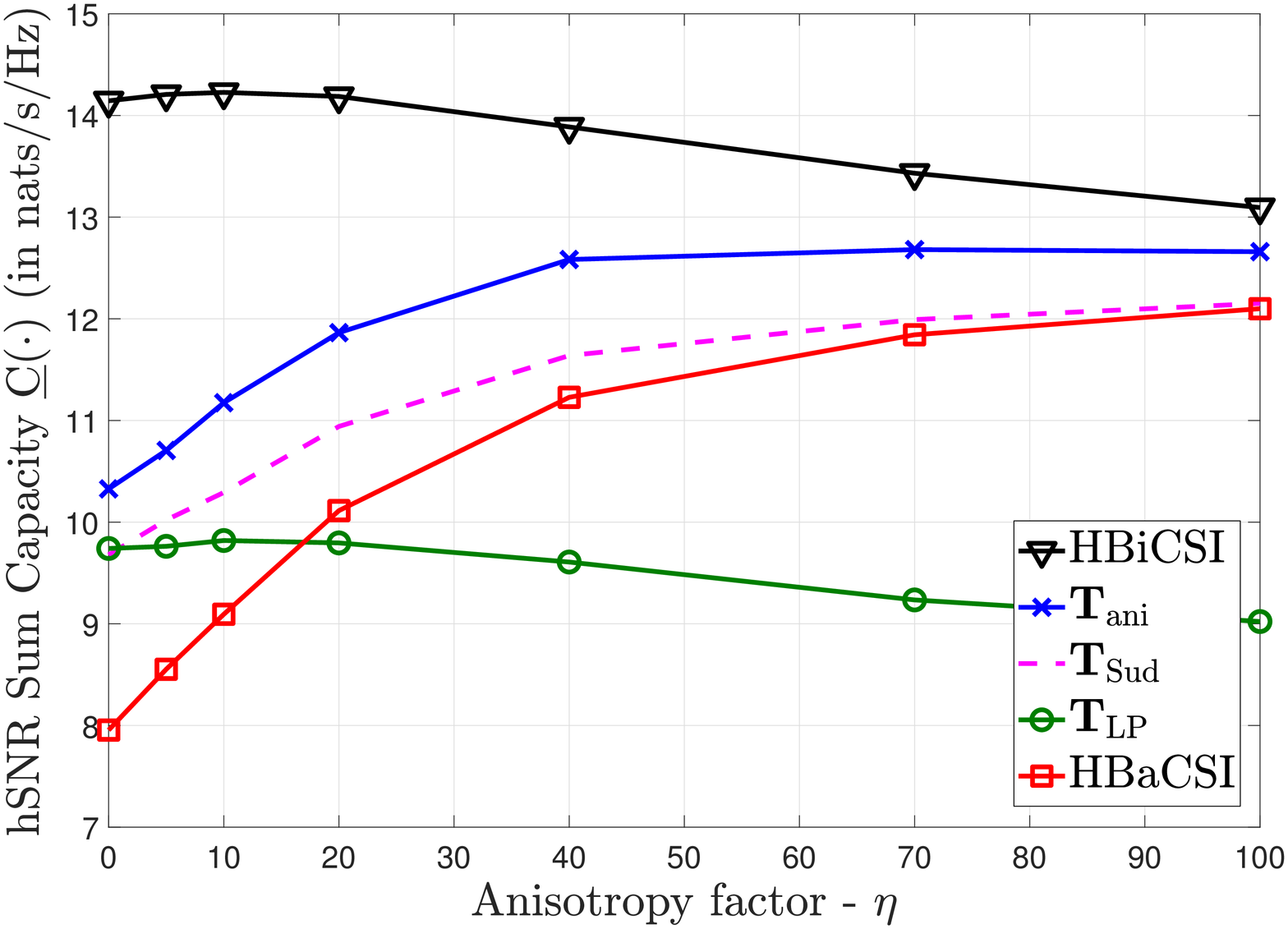}}
\caption{Ergodic hSNR sum capacity for the skewed beamformer $\mathbf{T}_{\rm ani}$ as a function of channel anisotropy. (a) $\hat{\mathbf{T}}_{\rm LP} = { \left[ \mathbf{W} \ \mathbb{O}_{(D-L)\times L} \right]}^{\dag}$, where $\mathbf{W}$ is the $L \times L$ DFT matrix (b) $\hat{\mathbf{T}}_{\rm LP}$ is from \cite{MedraRepository} \big(simulation parameters: $N=400$, $\mathcal{S} = \mathcal{S}_{\rm all}$, $\mathbf{R}_{\rm rx} = \mathbb{I}_M$ and $\rho = 1$\big).}
\label{fig_compare_anisotropic}
\end{figure}
For this channel, the hSNR sum-capacity of the skewed beamforming matrix $\mathbf{T}_{\rm ani}$ is compared to $\mathbf{T}_{\rm LP} = \mathbf{E}^D_{\rm tx} \hat{\mathbf{T}}_{\rm LP}$ in Fig.~\ref{fig_compare_anisotropic} as a function of the channel anisotropy, for the both cases of $L \leq D$ and $L > D$. For a comparison to existing designs, we also depict the hSNR capacity of $\mathbf{T}_{\rm Sud} = {[\mathbf{E}_{\rm tx}]}_{{\rm c}\{ \mu(1:L)\}}$, a generalization of the beamformer in \cite{Sudarshan}, where:\footnote{Here $\mathbf{T}_{\rm Sud}$ is a column permutation of ${[\mathbf{E}_{\rm tx}]}_{{\rm c}\{ 1:L\}}$ which ensures good performance with $\mathcal{S}_{\rm all}$.}
$$\mu(\ell) = \text{mod} \bigg( (\ell-1)K + \Big\lfloor \frac{(\ell-1)K}{L} \Big\rfloor, L\bigg) + 1. $$
We observe from the results that the hSNR capacity gap between HBiCSI and HBaCSI reduces as the channel anisotropy increases. Also, unlike $\mathbf{T}_{\rm LP}$, the hSNR sum-capacity of $\mathbf{T}_{\rm ani}$ does not reduce with increasing anisotropy. Further, we observe that $\mathbf{T}_{\rm ani}$ outperforms both $\mathbf{T}_{\rm Sud}$ and HBaCSI, over the whole range of $\eta$ for both $D \leq L$ and $D > L$. Therefore the designed beamformer is not only more generic (allows non-orthogonal columns) but also leads to a higher hSNR sum-capacity in comparison to existing designs.

\section{Conclusions} \label{sec_conclusions}
In this work we propose a generic architecture for HBwS, as an attractive solution to reduce complexity and cost of massive MIMO systems. We show that a beamforming matrix that maximizes a lower bound to the system sum capacity is obtained by solving a coupled Grassmannian sub-space packing problem. We propose a sub-optimal solution to this problem, and explore algorithms to improve the design for a given $\mathcal{S}$. 
We show that designing the beamformer apriori for a fixed $D$ enables a two stage implementation, having a fixed and an adaptive stage. This implementation lowers the hardware implementation cost significantly when $L > D$. We also show that switch positions with low overlap contribute more to hSNR sum capacity than others, and provide an algorithm for finding a family of such important switch positions. 
Simulation results suggest that with $L = 2D$, HBwS can achieve gains comparable to half the hSNR capacity gap between HBaCSI and HBiCSI. We also conclude that for the explored switch position sets, the RD-beamformer design $ \hat{\mathbf{T}}_{\rm LP}$ cannot be improved further via conventional approaches such as Algorithms \ref{algo_permute_mtx} \& \ref{algo_grad_descent}. 
Furthermore, for a good trade-off between performance and hardware cost, the number of input ports ($L$) should be of the order of $D$. However, larger values may be practical in a multi-user scenario, since a larger $L$ can aid separation of multiple data streams. In particular, it helps make performance of linear ZF precoding comparable to the non-linear and capacity optimal DPC, without the need for sophisticated user scheduling algorithms. Picking $D \ll K/\zeta$ ensures a low channel estimation overhead. Results for anisotropic channels suggest that skewing the line packed beamformer yields better performance in comparison to existing designs. 

As is the case with other such multi-antenna techniques \cite{Molisch_mag}, the performance gain of HBwS over HBaCSI decreases with frequency selective fading. This is because in the limit of a large transmission bandwidth, frequency diversity makes all the $L$ selection ports equivalent. Hence, HBwS is more suited for small-to-medium bandwidth systems, where channels are at-most moderately frequency selective.

\begin{appendices}
\section{} \label{appdix1}
\begin{proof}[(Proof of Theorem \ref{Th_restrict_dom_eigspace})]
Let $\mathbf{E}_{\rm tx}^D = {[\mathbf{E}_{\rm tx}]}_{{\rm c}\{1:D\}}$, and let us additionally define $\mathbf{E}_{\rm tx}^{N-D} \triangleq {[\mathbf{E}_{\rm tx}]}_{{\rm c}\{D+1:N\}}$. Since $\mathbf{E}_{\rm tx} \mathbf{E}_{\rm tx}^{\dag} = \mathbb{I}_N$, for any beamforming matrix $\mathbf{T}$ we can write:
\begin{eqnarray} \label{eqn_th1_T_split}
\mathbf{T} = \mathbf{E}_{\rm tx} \mathbf{E}_{\rm tx}^{\dag} \mathbf{T} = \mathbf{E}_{\rm tx}^{D} {[\mathbf{E}_{\rm tx}^{D}]}^{\dag} \mathbf{T} + \mathbf{E}_{\rm tx}^{N-D} {[\mathbf{E}_{\rm tx}^{N-D}]}^{\dag} \mathbf{T}.
\end{eqnarray}
Additionally let the principal $D \times D$ submatrix of $\boldsymbol{\Lambda}_{\rm tx}$, that contains all non-zero eigenvalues, be given by $\boldsymbol{\Lambda}^{D}_{\rm tx}$. Defining $\hat{\mathbf{T}} \triangleq {[\mathbf{E}_{\rm tx}^D]}^{\dag} \mathbf{T}$, for any user $m$, from \eqref{eqn_downlink_y1} and \eqref{eqn_channel_KL_exapnsion} we have:
\begin{eqnarray}
\mathbf{y}_m(\mathbf{S}_i) &=& \sqrt{\rho} \mathbf{R}_{{\rm rx},m}^{1/2} \mathbf{H}_m {[\boldsymbol{\Lambda}_{\rm tx}]}^{1/2} {[\mathbf{E}_{\rm tx}]}^{\dag} \mathbf{T} \mathbf{S}_i \mathbf{x} + \mathbf{n}_m \nonumber \\
&=& \sqrt{\rho} \mathbf{R}_{{\rm rx},m}^{1/2} \mathbf{H}_m {[\boldsymbol{\Lambda}^D_{\rm tx}]}^{1/2} {[\mathbf{E}^D_{\rm tx}]}^{\dag} \mathbf{E}_{\rm tx}^{D} \hat{\mathbf{T}} \mathbf{S}_i \mathbf{x} + \mathbf{n}_m, \label{eqn_optThat_th1}
\end{eqnarray}
where \eqref{eqn_optThat_th1} follows from \eqref{eqn_th1_T_split} and the fact that $\text{rank}\{\mathbf{R}_{\rm tx}\} \leq D$. Now, from the transmit power constraint in \eqref{eqn_TX_pow_constraint1} we have:
\begin{subequations}
\begin{flalign}
& \quad \ \ \text{tr} \left\{ \mathbf{T}\mathbf{S}_i \mathbb{E}_{\mathbf{x}}\{\mathbf{x} \mathbf{x}^{\dag}\} \mathbf{S}_i^{\dag} \mathbf{T}^{\dag} \right\} \leq 1 & \nonumber \\
& \Rightarrow \text{tr} \left\{ \mathbf{E}_{\rm tx}^D \Big( {[\mathbf{E}_{\rm tx}^D]}^{\dag} \mathbf{E}_{\rm tx}^D \Big) \hat{\mathbf{T}} \mathbf{S}_i \mathbb{E}_{\mathbf{x}}\{\mathbf{x} \mathbf{x}^{\dag}\} \mathbf{S}_i^{\dag} \mathbf{T}^{\dag}  \right\} & \nonumber \\
& \qquad \qquad + \text{tr} \left\{ \mathbf{E}_{\rm tx}^{N-D} {[\mathbf{E}_{\rm tx}^{N-D}]}^{\dag} \mathbf{T} \mathbf{S}_i \mathbb{E}_{\mathbf{x}}\{\mathbf{x} \mathbf{x}^{\dag}\} \mathbf{S}_i^{\dag} \mathbf{T}^{\dag} \right\} \leq 1 \!\!\! & \label{eqn_pow_const_th1_1}\\
& \Rightarrow \text{tr} \left\{ \mathbf{E}_{\rm tx}^D \hat{\mathbf{T}} \mathbf{S}_i \mathbb{E}_{\mathbf{x}}\{\mathbf{x} \mathbf{x}^{\dag}\} \mathbf{S}_i^{\dag} \hat{\mathbf{T}}^{\dag} {[\mathbf{E}_{\rm tx}^D]}^{\dag} \right\} & \nonumber \\
& \qquad \qquad + \text{tr} \left\{ {[\mathbf{E}_{\rm tx}^{N-D}]}^{\dag}\mathbf{T} \mathbf{S}_i \mathbb{E}_{\mathbf{x}}\{\mathbf{x} \mathbf{x}^{\dag}\} \mathbf{S}_i^{\dag} \mathbf{T}^{\dag} \mathbf{E}_{\rm tx}^{N-D} \right\} \leq 1 \!\!\! & \label{eqn_pow_const_th1_2} \\
& \Rightarrow \text{tr} \left\{ \mathbf{E}_{\rm tx}^{D} \hat{\mathbf{T}} \mathbf{S}_i \mathbb{E}_{\mathbf{x}}\{\mathbf{x} \mathbf{x}^{\dag}\} \mathbf{S}_i^{\dag} \hat{\mathbf{T}}^{\dag} {[\mathbf{E}_{\rm tx}^{D}]}^{\dag} \right\} \leq 1, & \label{eqn_pow_const_th1_3}
\end{flalign}
\end{subequations}
where \eqref{eqn_pow_const_th1_1} follows by replacing $\mathbf{T}$ using \eqref{eqn_th1_T_split} and by using ${[\mathbf{E}_{\rm tx}^D]}^{\dag} \mathbf{E}_{\rm tx}^D = \mathbb{I}_D$, \eqref{eqn_pow_const_th1_2} follows from the identity that $\text{tr}\{\mathbf{A^{\dag}B}\} = \text{tr}\{\mathbf{BA^{\dag}}\}$ for matrices $\mathbf{A,B}$ of same dimension, and \eqref{eqn_pow_const_th1_3} follows by observing that both matrices on the left hand side are positive semi-definite.
From \eqref{eqn_pow_const_th1_3} and \eqref{eqn_optThat_th1}, for any $\mathbf{T}$ and $\mathbb{E}\{\mathbf{x} \mathbf{x}^{\dag}\}$, there exists a $D\times L$ matrix $\hat{\mathbf{T}}$ such that if $\mathbf{T}$ satisfies the power constraint then so does $\mathbf{E}_{\rm tx}^D \hat{\mathbf{T}}$ and both $\mathbf{T}, \mathbf{E}_{\rm tx}^D \hat{\mathbf{T}}$ yield the same instantaneous received signal $\mathbf{y}_m(\mathbf{S}_i)$. Therefore the theorem follows. 
\end{proof}

\section{} \label{appdix2}
\begin{proof}[(Proof of Theorem \ref{Th_bound_eigspace})]
Let $\hat{\mathbf{T}} \in \mathbb{C}^{D \times L}$ be any $D \times L$ matrix. For each $\hat{\mathbf{T}} \mathbf{S}_i$ we have a corresponding $K \times K$ ortho-normalization matrix $\hat{\mathbf{G}}_i$. 
Now consider $\hat{\mathbf{T}}_{\theta} = \hat{\mathbf{T}} \boldsymbol{\Lambda}_{\theta}$ where $\boldsymbol{\Lambda}_{\theta}$ is any $L \times L$ non-singular complex diagonal matrix. For each selection matrix $\mathbf{S}_i$, by defining a corresponding matrix $\hat{\mathbf{G}}_{i\theta} \triangleq \mathbf{S}_i^{\dag} \boldsymbol{\Lambda}_{\theta}^{-1} \mathbf{S}_i \hat{\mathbf{G}}_i$, we have:
\begin{subequations}
\begin{flalign}
& \hat{\mathbf{G}}_{i\theta}^{\dag} \mathbf{S}_i^{\dag} \hat{\mathbf{T}}_{\theta}^{\dag} \hat{\mathbf{T}}_{\theta} \mathbf{S}_i \hat{\mathbf{G}}_{i\theta} = \hat{\mathbf{G}}_i^{\dag} \mathbf{S}_i^{\dag} {[\boldsymbol{\Lambda}_{\theta}^{-1}]}^{\dag} \mathbf{S}_i \mathbf{S}_i^{\dag} \boldsymbol{\Lambda}_{\theta}^{\dag} \hat{\mathbf{T}}^{\dag} \hat{\mathbf{T}} \boldsymbol{\Lambda}_{\theta} \mathbf{S}_i \mathbf{S}_i^{\dag} \boldsymbol{\Lambda}_{\theta}^{-1} \mathbf{S}_i \hat{\mathbf{G}}_i & \nonumber \\
&\qquad \qquad \qquad \quad = \hat{\mathbf{G}}_i^{\dag} \mathbf{S}_i^{\dag} \mathbf{S}_i \mathbf{S}_i^{\dag} \hat{\mathbf{T}}^{\dag} \hat{\mathbf{T}} \mathbf{S}_i \mathbf{S}_i^{\dag} \mathbf{S}_i \hat{\mathbf{G}}_i & \label{eqn_th2_2ndlst_step} \\
&\qquad \qquad \qquad \quad = \hat{\mathbf{G}}_i^{\dag} \mathbf{S}_i^{\dag} \hat{\mathbf{T}}^{\dag} \hat{\mathbf{T}} \mathbf{S}_i \hat{\mathbf{G}}_i & \label{eqn_th2_lst_step} \\
&\qquad \qquad \qquad \quad = \mathbb{I}_{K}, & \nonumber 
\end{flalign}
\end{subequations}
where \eqref{eqn_th2_2ndlst_step} follows from the fact that $\mathbf{S}_i\mathbf{S}_i^{\dag}$ is diagonal and hence commutes with $\boldsymbol{\Lambda}_{\theta}$ and \eqref{eqn_th2_lst_step} uses the fact that $\mathbf{S}_i^{\dag} \mathbf{S}_i = \mathbb{I}_K$. This proves that $\hat{\mathbf{G}}_{i\theta}$ ortho-normalizes columns of $\hat{\mathbf{T}}_{\theta} \mathbf{S}_i$. 

Using a similar sequence of steps it can be shown that $\hat{\mathbf{T}}_{\theta} \mathbf{S}_i \hat{\mathbf{G}}_{i\theta} \hat{\mathbf{G}}_{i\theta}^{\dag} \mathbf{S}_i^{\dag} \hat{\mathbf{T}}_{\theta}^{\dag} = \hat{\mathbf{T}} \mathbf{S}_i \hat{\mathbf{G}}_i \hat{\mathbf{G}}_i^{\dag} \mathbf{S}_i^{\dag} \hat{\mathbf{T}}^{\dag}$ and hence from \eqref{eqn_def_fobjD}, $\underbar{C}^{D}(\hat{\mathbf{T}}_{\theta}) = \underbar{C}^{D}(\hat{\mathbf{T}})$. This proves the theorem.
\end{proof}
\section{} \label{appdix3}
\begin{proof}[(Proof of Lemma \ref{Th_subopt_selection})]
Consider the decomposition $\mathbf{H}^D = \mathbf{U}_{\mathbf{H}^D} \boldsymbol{\Lambda}_{\mathbf{H}^D} \mathbf{V}^{\dag}_{\mathbf{H}^D}$, where $\mathbf{U}_{\mathbf{H}^D}$ is the $M\times M$ left singular-vector matrix, $ \boldsymbol{\Lambda}_{\mathbf{H}^D}$ is the $M\times M$ diagonal matrix containing the non-zero singular values, and $\mathbf{V}_{\mathbf{H}^D}$ is the $D \times M$ semi-unitary matrix containing the $M$ right singular vectors (corresponding to the non-zero singular values) for $\mathbf{H}^D$ in \eqref{eqn_def_fobjD}.\footnote{Note that this is the \emph{compact singular value decomposition} \cite{Edelman1999} which is slightly different from the usual approach of singular value decomposition.} Now, using the fact that $|\mathbb{I} + \mathbf{A}| \geq 1+ |\mathbf{A}|$ for any positive semi-definite matrix $\mathbf{A}$, we can bound $\underbar{C}^{D}(\hat{\mathbf{T}})$ in \eqref{eqn_def_fobjD}, for $\boldsymbol{\Lambda}^D_{\rm tx} = \mathbb{I}_D$, as:
\begin{eqnarray}
\underbar{C}^{D}(\hat{\mathbf{T}}) \geq \mathbb{E}_{\boldsymbol{\Lambda}_{\mathbf{H}^D}} \mathbb{E}_{\mathbf{V}_{\mathbf{H}^D}} \left\{ \max_{1 \leq i \leq |\mathcal{S}|} \log \left[ 1 + \alpha \left| \mathbf{V}^{\dag}_{\mathbf{H}^D} \mathbf{Q}_i \mathbf{Q}^{\dag}_i \mathbf{V}_{\mathbf{H}^D} \right| \right] \right\}, \label{eqn_cap_LB1} 
\end{eqnarray}
where $\alpha = {\left(\frac{\rho}{M} \right)}^M \left|\mathbf{R}_{\rm rx} \right| {\left| \boldsymbol{\Lambda}_{\mathbf{H}^D} \right|}^2$. Since $\mathbf{H}^D$ has i.i.d. $\mathcal{CN}(0,1)$ components, it is well known that $\boldsymbol{\Lambda}_{\mathbf{H}^D}$ and $\mathbf{V}_{\mathbf{H}^D}$ are independently distributed and further, $\mathbf{V}_{\mathbf{H}^D}$ is uniformly distributed over $\mathcal{U}(D,M)$ \cite[Lemma 4]{Love_SMux}. 

Consider a random $D\times K$ matrix $\mathbf{V}$, independent of $\mathbf{H}^D$ and uniformly distributed over $\mathcal{U}(D,K)$. Since the uniform measure is invariant under unitary transformation, for any $D \times D$ unitary matrix $\mathbf{U}$, we have:
$$\mathbf{V} \,{\buildrel d \over =}\, \mathbf{U V} \Rightarrow {[\mathbf{V}]}_{{\rm c}\{1:M\}} \,{\buildrel d \over =}\, \mathbf{U} {[\mathbf{V}]}_{{\rm c}\{1:M\}}.$$
In other words, ${[\mathbf{V}]}_{{\rm c}\{1:M\}}$ is uniformly distributed over $\mathcal{U}(D,M)$ and therefore $\mathbf{V}_{\mathbf{H}^D} \,{\buildrel d \over =}\, {[\mathbf{V}]}_{{\rm c}\{1:M\}}$. Now replacing $\mathbf{V}_{\mathbf{H}^D}$ in \eqref{eqn_cap_LB1} by ${[\mathbf{V}]}_{{\rm c}\{1:M\}}$, we have:
\begin{flalign}
& \underbar{C}^{D}(\hat{\mathbf{T}}) \geq & \nonumber \\
& \quad \mathbb{E}_{\boldsymbol{\Lambda}_{\mathbf{H}^D}} \mathbb{E}_{\mathbf{V}} \! \left\{ \max_{1 \leq i \leq |\mathcal{S}|} \!\! \log \! \left[1 \!+\! \alpha \! \left| {[\mathbf{V}]}_{{\rm c}\{1:M\}}^{\dag} \mathbf{Q}_i \mathbf{Q}^{\dag}_i {[\mathbf{V}]}_{{\rm c}\{1:M\}} \right| \right] \right\} \!\!\!\!\! & \label{eqn_th3_proof1}
\end{flalign}
Note that ${[\mathbf{V}]}_{{\rm c}\{1:M\}}^{\dag} \mathbf{Q}_i \mathbf{Q}^{\dag}_i {[\mathbf{V}]}_{{\rm c}\{1:M\}}$ is the $M\times M$ principal sub-matrix of $\mathbf{V}^{\dag} \mathbf{Q}_i \mathbf{Q}^{\dag}_i \mathbf{V}$. Therefore we have $\forall 1 \leq i\leq M$:
\begin{subequations} \label{eqn_th3_proof2}
\begin{flalign}
& \lambda^{\downarrow}_i\{ {[\mathbf{V}]}_{{\rm c}\{1:M\}}^{\dag} \mathbf{Q}_i \mathbf{Q}^{\dag}_i {[\mathbf{V}]}_{{\rm c}\{1:M\}} \} \geq \lambda^{\downarrow}_{i+(K-M)}\{ \mathbf{V}^{\dag} \mathbf{Q}_i \mathbf{Q}^{\dag}_i \mathbf{V} \}, \!\!\!\!\!\!\!\!  & \label{eqn_th3_proof2a} \\
& 1 \geq \lambda^{\downarrow}_{k}\{ \mathbf{V}^{\dag} \mathbf{Q}_i \mathbf{Q}^{\dag}_i \mathbf{V} \} \geq 0 \ \ \forall 1 \leq k \leq K, &  \label{eqn_th3_proof2b}
\end{flalign}
\end{subequations}
where $\lambda^{\downarrow}_k\{\mathbf{A}\}$ represents the $k$-th largest eigenvalue of a square matrix $\mathbf{A}$, \eqref{eqn_th3_proof2a} follows from the Cauchy's Interlacing Theorem \cite[Corollary 3.1.5]{Bhatia} and \eqref{eqn_th3_proof2b} is obtained by using the results on the spectral norm \cite{Roger1990} and by observing that both $\mathbf{V}$ and $\mathbf{Q}_i$ have a largest singular value of $1$. Since the determinant is the product of eigenvalues, from \eqref{eqn_th3_proof1} and \eqref{eqn_th3_proof2}, we arrive at \eqref{eqn_th3_LB2}. Note that ${\left| \boldsymbol{\Lambda}_{\mathbf{H}^D} \right|}^2 = \big|\mathbf{H}^D {[\mathbf{H}^D]}^{\dag} \big|$.
\end{proof}

\section{} \label{appdix5}
Here we explore some numerical algorithms to find good solutions to \eqref{eqn_opt_prbm_Grassmann}, starting with an initial solution of $\hat{\mathbf{T}}_{\rm LP}$. First we propose the use of a greedy algorithm that permutes columns of $\hat{\mathbf{T}}_{\rm LP}$ to increase $f_{\rm FS}(\cdot)$, as depicted in Algorithm \ref{algo_permute_mtx}. 
The performance of this permuted matrix $\hat{\mathbf{T}}_{\rm Alg2}$ may further be improved via a numerical gradient ascent of $f_{\rm FS}(\hat{\mathbf{T}})$, as depicted in Algorithm \ref{algo_grad_descent}. Since $f_{\rm FS}(\hat{\mathbf{T}})$ involves the $\min\{\}$ function which may not be differentiable, we use a differentiable relaxation of $f_{\rm FS}(\hat{\mathbf{T}})$ in Algorithm \ref{algo_grad_descent}:
$$\tilde{f}_{\rm FS}(\hat{\mathbf{T}}) = - \log \bigg[ \sum_{i}^{|\mathcal{S}|} \sum_{j = i+1}^{|\mathcal{S}|} \exp{ \big\{ -10 d_{\rm FS}(\mathbf{Q}_i,\mathbf{Q}_j) \big\}} \bigg]/10. $$
We therefore have $\hat{\mathbf{T}}_{\rm LP} \,{\buildrel \text{Algo\ref{algo_permute_mtx}} \over \longrightarrow }\, \hat{\mathbf{T}}_{\rm Alg2} \,{\buildrel \text{Algo\ref{algo_grad_descent}} \over \longrightarrow }\, \hat{\mathbf{T}}_{\rm Alg3}$, where: 
$$ f_{\rm FS}\big(\hat{\mathbf{T}}_{\rm LP} \big) \leq f_{\rm FS} \big(\hat{\mathbf{T}}_{\rm Alg2} \big) \leq f_{\rm FS} \big(\hat{\mathbf{T}}_{\rm Alg3} \big) \leq f_{\rm FS} \big(\hat{\mathbf{T}}_{\rm FS} \big).$$
\begin{algorithm}
\caption{Greedy column permutation algorithm}\label{algo_permute_mtx}
\begin{algorithmic}[1]
\STATE Inputs: $D, L, \mathcal{S}$
\STATE Initialize $\hat{\mathbf{T}}$
\COMMENT{As obtained by line packing \cite{Vishnu_ICC2017}}
\REPEAT
\FOR{$\ell = 1$ to $L$}
\FOR{$j = 1$ to $L$}
\STATE{$\hat{\mathbf{T}}_{\rm Alg2}(j) = \hat{\mathbf{T}}$}
\STATE{${[\hat{\mathbf{T}}_{\rm Alg2}(j)]}_{{\rm c}\{\ell\}} = {[\hat{\mathbf{T}}]}_{{\rm c}\{j\}}$ and ${[\hat{\mathbf{T}}_{\rm Alg2}(j)]}_{{\rm c}\{j\}} = {[\hat{\mathbf{T}}]}_{{\rm c}\{\ell\}}$}
\COMMENT{Swap columns $\ell$ and $j$}
\ENDFOR
\STATE{Find $j^{*}$ such that $f_{\rm FS}(\hat{\mathbf{T}}_{\rm Alg2}(j^{*}))$ is largest}
\IF{$f_{\rm FS}(\hat{\mathbf{T}}_{\rm Alg2}(j^{*})) > f_{\rm FS}(\hat{\mathbf{T}})$}
\STATE{$\hat{\mathbf{T}} = \hat{\mathbf{T}}_{\rm Alg2}(j^{*})$}
\ENDIF
\ENDFOR
\UNTIL{Convergence of $f_{\rm FS}(\hat{\mathbf{T}})$}
\RETURN $\hat{\mathbf{T}}$
\end{algorithmic}
\end{algorithm}
\begin{algorithm}
\caption{Gradient ascent algorithm for RD-beamformer design}\label{algo_grad_descent}
\begin{algorithmic}[1]
\STATE Inputs: $D, L, \mathcal{S}$
\STATE Initialize $\hat{\mathbf{T}}, \gamma, \Delta T$
\STATE $\hat{\mathbf{T}}_{\rm new} = \hat{\mathbf{T}}$
\REPEAT
\STATE $\hat{\mathbf{T}} = \hat{\mathbf{T}}_{\rm new}$
\FOR{$i = 1$ to $D$}
\FOR{$j = 1$ to $L$}
\STATE $\hat{\mathbf{T}}_{\Delta} = \hat{\mathbf{T}}$
\STATE ${[\hat{\mathbf{T}}_{\Delta}]}_{ij} = {[\hat{\mathbf{T}}]}_{ij} + \Delta T$
\STATE $\mathbf{F}_{ij} = \frac{\tilde{f}_{\rm FS}(\hat{\mathbf{T}}_{\Delta}) - \tilde{f}_{\rm FS}(\hat{\mathbf{T}})}{\Delta T}$
\COMMENT{Computing the gradient of objective function}
\ENDFOR
\ENDFOR
\FOR{$j = 1$ to $L$}
\STATE $\mathbf{t} = {[\mathbf{F}]}_{{\rm c}\{j\}} - {[\hat{\mathbf{T}}]}_{{\rm c}\{j\}} {[\hat{\mathbf{T}}]}_{{\rm c}\{j\}}^{\dag} {[\mathbf{F}]}_{{\rm c}\{j\}} $
\COMMENT{Component of the gradient tangential to unit norm constraint}
\STATE ${[\hat{\mathbf{T}}_{\rm new}]}_{{\rm c}\{j\}} = {[\hat{\mathbf{T}}]}_{{\rm c}\{j\}} \cos(\gamma|\mathbf{t}|) + \mathbf{t}\sin(\gamma|\mathbf{t}|)/|\mathbf{t}|$
\COMMENT{Update process ensures unit norm columns \cite{Edelman1999}}
\ENDFOR
\UNTIL{$\tilde{f}_{\rm FS}(\hat{\mathbf{T}}_{\rm new}) \leq \tilde{f}_{\rm FS}(\hat{\mathbf{T}})$ }
\RETURN $\hat{\mathbf{T}}$
\end{algorithmic}
\end{algorithm}
The complexities of Algorithms \ref{algo_permute_mtx} and \ref{algo_grad_descent} are both $O(D L {|\mathcal{S}|}^2)$ per iteration, which can be substantial when $|\mathcal{S}|$ and/or $L$ is large. 

\section{} \label{appdix4}
Let us define for each selection matrix $\mathbf{S}_i$ a corresponding set $\mathcal{B}_i \subset \{1,..,L\}$ such that $\ell \in \mathcal{B}_i$ iff ${[\mathbf{S}_i]}_{{\rm c}\{k\}} = {[\mathbb{I}_L]}_{{\rm c}\{\ell\}}$ for some $1 \leq k \leq K$. Then for any $\ell \in \mathcal{B}_i \cap \mathcal{B}_j$, we have ${[\mathbf{\hat{T}}]}_{{\rm c}\{\ell\}} \in \mathcal{P}\{\mathbf{Q}_i\} \cap \mathcal{P}\{\mathbf{Q}_j\}$ i.e., there exists a $K\times 1$ vectors $\mathbf{i}_{\ell}, \mathbf{j}_{\ell}$ such that $\mathbf{Q}_i \mathbf{i}_{\ell} = \mathbf{Q}_j \mathbf{j}_{\ell} = {[\mathbf{\hat{T}}]}_{{\rm c}\{\ell\}}$. Furthermore since $\mathbf{Q}_i^{\dag} \mathbf{Q}_i = \mathbf{Q}_j^{\dag} \mathbf{Q}_j = \mathbb{I}_K$, we have:
\begin{eqnarray}
(\mathbf{Q}_j \mathbf{Q}_j^{\dag}) {[\mathbf{\hat{T}}]}_{{\rm c}\{\ell\}} &=& (\mathbf{Q}_j \mathbf{Q}_j^{\dag} ) \mathbf{Q}_j \mathbf{j}_{\ell} \nonumber \\
\Rightarrow \mathbf{Q}_j \mathbf{Q}^{\dag}_j {[\mathbf{\hat{T}}]}_{{\rm c}\{\ell\}} &=& {[\mathbf{\hat{T}}]}_{{\rm c}\{\ell\}} \nonumber \\
\Rightarrow \mathbf{Q}_j \mathbf{Q}^{\dag}_j \big( \mathbf{Q}_i \mathbf{i}_{\ell} \big) &=& \mathbf{Q}_i \mathbf{i}_{\ell} \nonumber \\
\Rightarrow \mathbf{Q}_i^{\dag} \mathbf{Q}_j \mathbf{Q}^{\dag}_j \big( \mathbf{Q}_i \mathbf{i}_{\ell} \big) &=& \mathbf{i}_{\ell}, \nonumber
\end{eqnarray}
i.e., $\mathbf{i}_{\ell}$ is an eigen-vector of $\mathbf{Q}^{\dag}_i \mathbf{Q}_j \mathbf{Q}^{\dag}_j \mathbf{Q}_i$ with eigenvalue $1$. If the vectors $\{{[\mathbf{\hat{T}}]}_{{\rm c}\{\ell\}} | \ell \in \mathcal{B}_i \cap \mathcal{B}_j\} $ are linearly independent, then so are $\{\mathbf{i}_{\ell} | \ell \in \mathcal{B}_i \cap \mathcal{B}_j\}$. Therefore, $\mathbf{Q}^{\dag}_i \mathbf{Q}_j \mathbf{Q}^{\dag}_j \mathbf{Q}_i$ has $|\mathcal{B}_i \cap \mathcal{B}_j|$ unity eigenvalues.
\end{appendices}

\bibliographystyle{ieeetr}
\bibliography{references_jrnl_ppr2}

\begin{IEEEbiography}[{\includegraphics[width=1in,height=1.25in,clip,keepaspectratio]{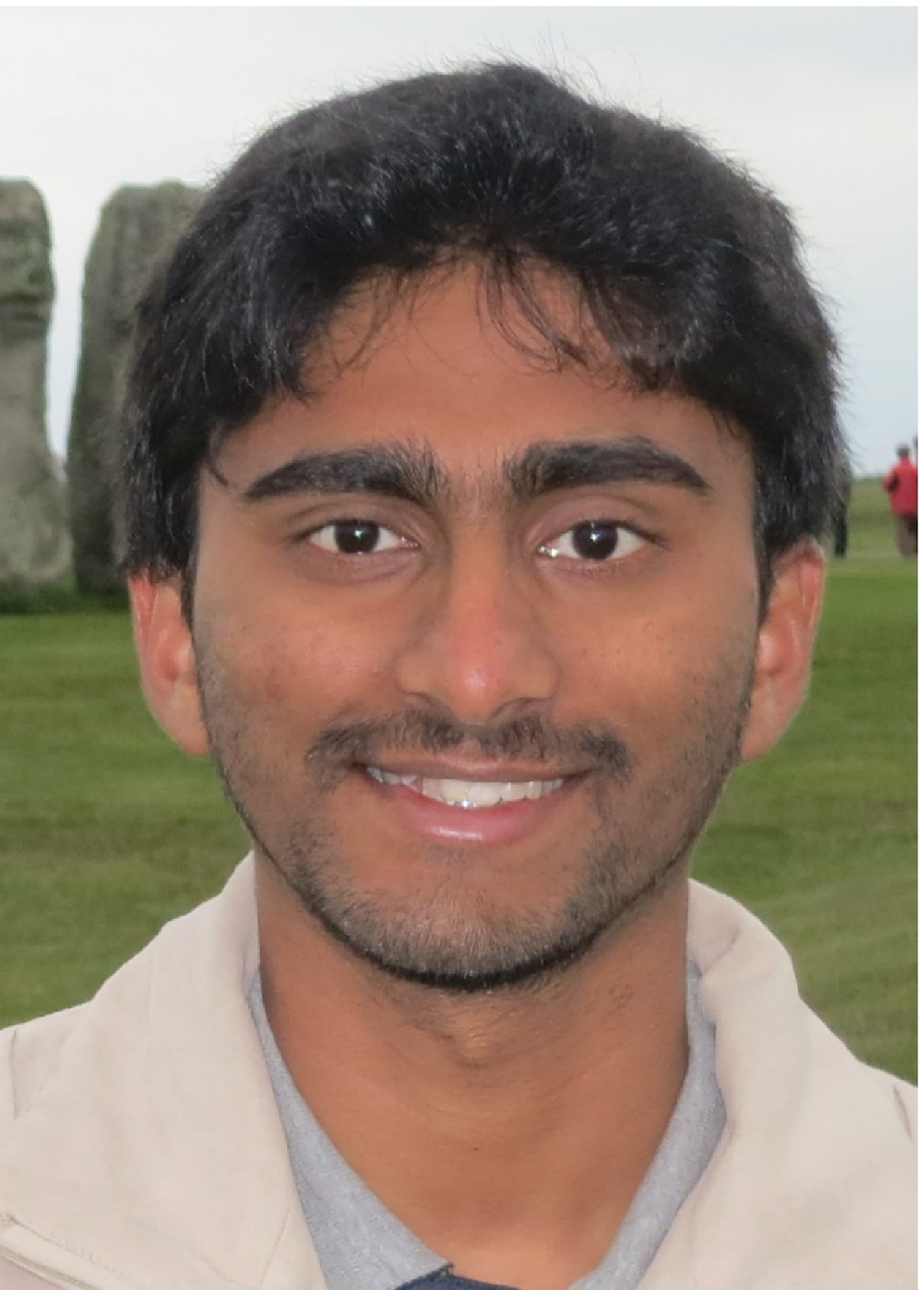}}]{Vishnu V. Ratnam}
(S'10) received the B.Tech. degree (Hons.) in electronics and electrical communication engineering from IIT Kharagpur, Kharagpur,
in 2012. He graduated as the Salutatorian for the class of 2012. He is currently pursuing the Ph.D. degree in electrical engineering with the University of Southern California. His research interests are in reduced complexity transceivers for large antenna systems (massive MIMO/mm-wave) and ultrawideband systems, channel estimation techniques, manifold signal processing and in resource allocation problems in multi-antenna networks. 

Mr. Ratnam is a recipient of the Best Student Paper Award at the IEEE International Conference on Ubiquitous Wireless Broadband (ICUWB) in 2016, and a member of the Phi-Kappa-Phi honor society.
\end{IEEEbiography}

\begin{IEEEbiography}[{\includegraphics[width=1in,height=1.25in,clip,keepaspectratio]{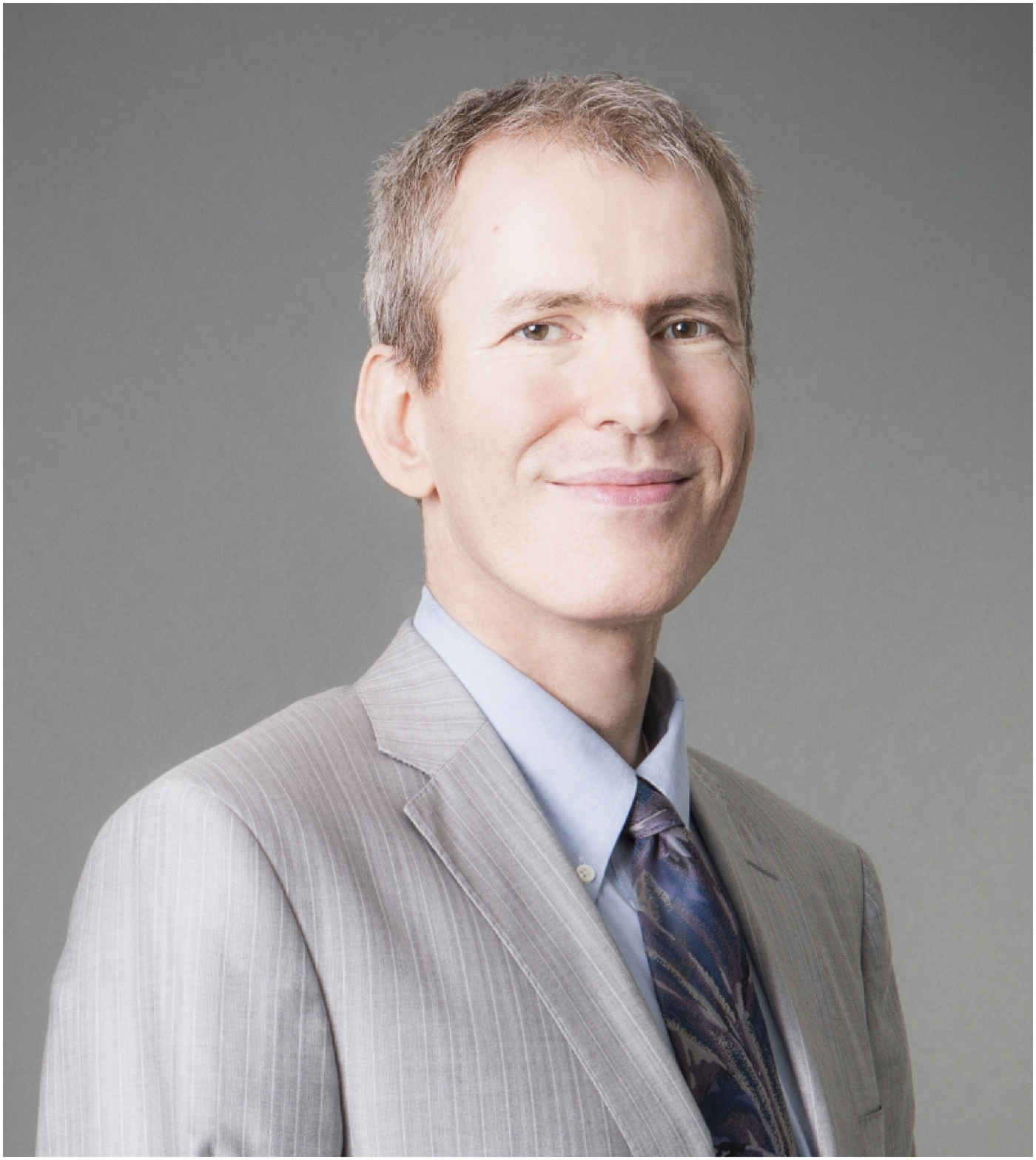}}]{Andreas F. Molisch}
(S'89--M'95--SM'00--F'05) received the Dipl. Ing., Ph.D., and habilitation degrees from the Technical University of Vienna, Vienna, Austria, in 1990, 1994, and 1999, respectively. He subsequently was with AT\&T (Bell) Laboratories Research (USA); Lund University, Lund, Sweden, and Mitsubishi Electric Research Labs (USA). He is now a Professor and the Solomon-Golomb -- Andrew-and-Erna Viterbi Chair at the University of Southern California, Los Angeles. 

His current research interests are the measurement and modeling of mobile radio channels, multi-antenna systems, wireless video distribution, ultra-wideband communications and localization, and novel modulation formats. He has authored, coauthored, or edited four books (among them the textbook Wireless Communications, Wiley-IEEE Press), 19 book chapters, more than 220  journal papers, more than 300 conference papers, as well as more than 80 patents and 70 standards contributions.

Dr. Molisch has been an Editor of a number of journals and special issues, General Chair, Technical Program Committee Chair, or Symposium Chair of multiple international conferences, as well as Chairman of various international standardization groups. He is a Fellow of the National Academy of Inventors, Fellow of the AAAS, Fellow of the IET, an IEEE Distinguished Lecturer, and a member of the Austrian Academy of Sciences. He has received numerous awards, among them the Donald Fink Prize of the IEEE, the IET Achievement Medal, and the Eric Sumner Award of the IEEE.
\end{IEEEbiography}

\begin{IEEEbiography}[{\includegraphics[width=1in,height=1.25in,clip,keepaspectratio]{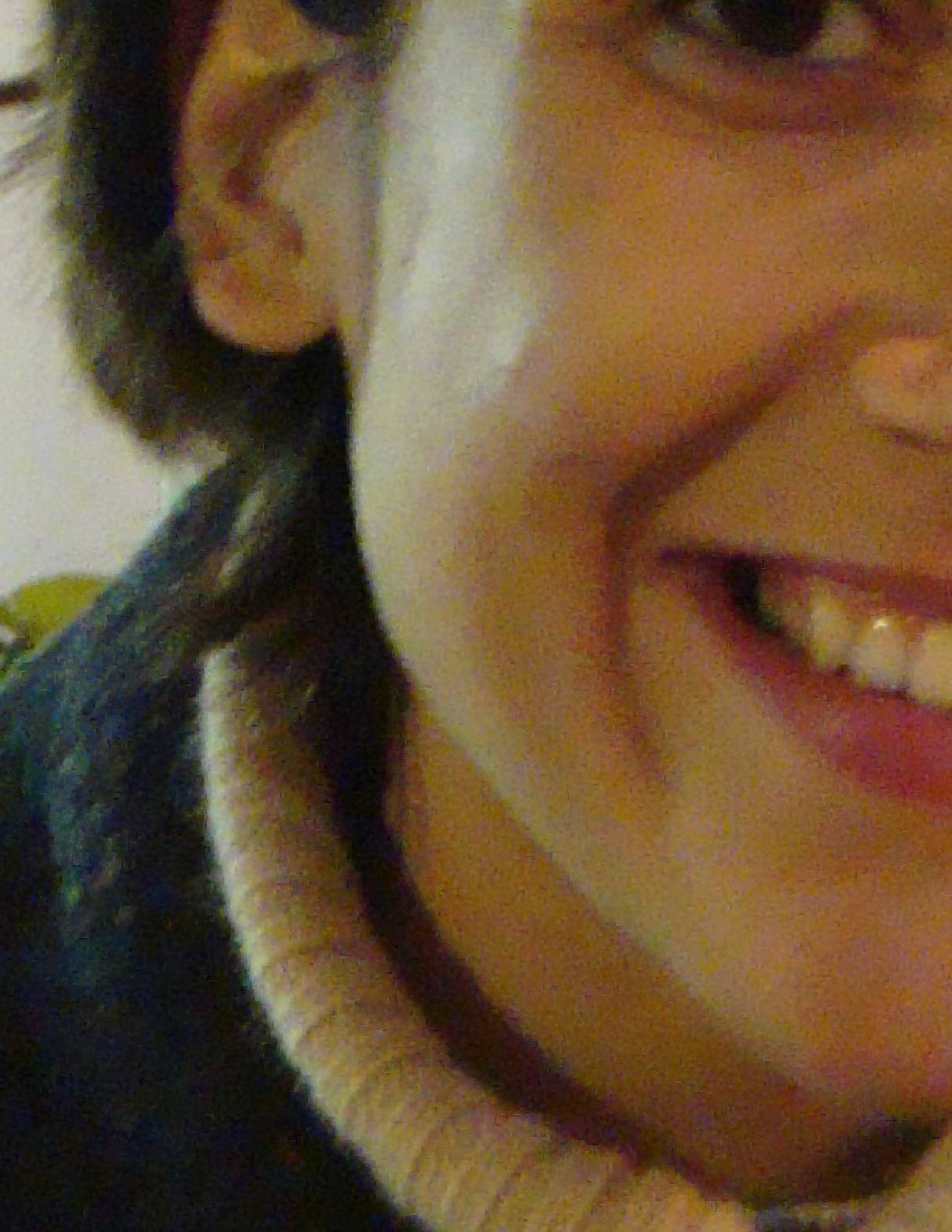}}]{Ozgun Bursalioglu Yilmaz}
(M'12) is a Research Engineer at Docomo Innovations Inc., working in the area of wireless communications on MIMO techniques and LTE enhancements since 2012. She graduated from the Ming Hsieh Department of Electrical Engineering, University of Southern California, in 2011.  Her Ph.D. thesis is on joint source channel coding for multicast and multiple description coding scenarios using rateless codes. Previously she received M.S. and B.S. degrees from University of California, Riverside (2006) and Middle East Technical University (METU), Ankara, Turkey (2004), respectively. She received the best student paper award at the International Conference on Acoustics, Speech and Signal Processing (IEEE ICASSP), in 2006.
\end{IEEEbiography}

\begin{IEEEbiography}[{\includegraphics[width=1in,height=1.25in,clip,keepaspectratio]{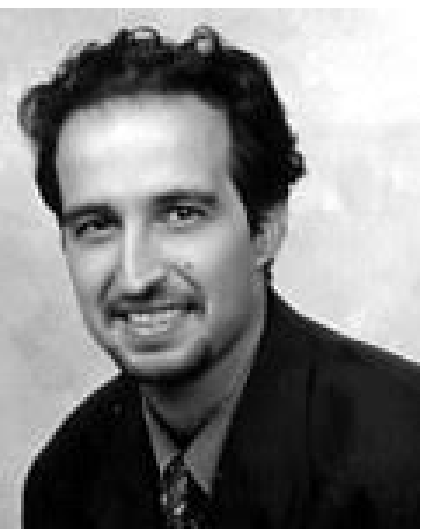}}]{Haralabos C. Papadopoulos} (S'92--M'98) received the S.B., S.M., and Ph.D. degrees from the Massachusetts Institute of Technology, Cambridge, MA, all in electrical engineering and computer science, in 1990, 1993, and 1998, respectively.

Since December 2005, he has been with DOCOMO Innovations, Palo Alto, CA,  working on physical-layer algorithms for wireless communication systems and architectures. From 1998 to 2005, he was on the faculty of the Department of Electrical and Computer Engineering, University of Maryland, College Park, MD, and held a joint appointment with the Institute of Systems Research. During his 1993--1995 summer visits to AT\&T Bell Labs, Murray Hill, NJ, he worked on shared time-division duplexing systems and digital audio broadcasting.  His research interests are in the areas of communications and signal processing, with emphasis on resource-efficient algorithms and architectures for wireless communication systems.

Dr. Papadopoulos is the recipient of an NSF CAREER Award (2000), the G. Corcoran Award (2000) given by the University of Maryland, College Park, and the 1994 F. C. Hennie Award (1994) given by the MIT EECS department. He is also a coauthor of the VTC Fall 2009 Best Student Paper Award.  He is a member of Eta Kappa Nu and Tau Beta Pi. He is also active in the industry and an inventor on several issued and pending patents. 
\end{IEEEbiography}

\end{document}